\documentclass[journal]{IEEEtran}
\IEEEoverridecommandlockouts

\title{Wideband Illumination with Liquid Crystal Reconfigurable Intelligent Surfaces: Modeling, Design, and Experimental Tests}

\author{
\IEEEauthorblockN{Mohamadreza Delbari$^{\orcidlink{0000-0002-4768-5874}}$, %\textit{Student Member, IEEE,} 
Robin Neuder$^{\orcidlink{0000-0002-2909-7841}}$, %\textit{Student Member, IEEE,} 
Alejandro Jim\'{e}nez-S\'{a}ez$^{\orcidlink{0000-0003-0468-1352}}$, %\textit{Member, IEEE,} 
Qikai Zhou$^{\orcidlink{0009-0005-3434-9483}}$, %\textit{Senior Member, IEEE,} 
and Vahid Jamali$^{\orcidlink{0000-0003-3920-7415}}$ %\textit{Senior Member, IEEE}
}
\IEEEauthorblockA{
\thanks{The work of M. Delbari, Q. Zhou, and V. Jamali was supported in part by the Deutsche Forschungsgemeinschaft (DFG, German Research Foundation) within the Collaborative Research Center Multi-Mechanisms Adaptation for the Future Internet (MAKI) (SFB 1053) under Project-ID 210487104; in part by the LOEWE initiative (Hesse, Germany) within the emergenCITY Centre under Grant LOEWE/1/12/519/03/05.001(0016)/72, and in part by the German Federal Ministry for Research, Technology and Space (BMFTR) under the program of ``Souverän. Digital. Vernetzt.'' joint project Open6GHub plus (Project-ID 16KIS2407). Neuder and Jim\'{e}nez-S\'{a}ez's work was supported by the Deutsche Forschungsgemeinschaft (DFG, German Research Foundation) – Project-ID 287022738 – TRR 196 MARIE within project C09. In addition, thanks goes to Merck Electronics KGaA, Darmstadt, Germany, for providing the liquid crystal mixture. This paper was presented in part at the 2025 IEEE Global Communications Workshops (Globecom, workshops) 
in \cite{delbari2025wideband}. (\textit{Corresponding author: Mohamadreza Delbari.})\\
Mohamadreza Delbari, Qikai Zhou and Vahid Jamali are with the Resilient Communication Systems Laboratory, Technische Universität Darmstadt, 64283 Darmstadt, Germany (e-mail: mohamadreza.delbari@tu-darmstadt.de; qikai.zhou@stud.tu-darmstadt.de;
 vahid.jamali@tu-darmstadt.de).
Robin Neuder and Alejandro Jim\'{e}nez-S\'{a}ez are with the Institute of Microwave Engineering and Photonics, Technische Universität Darmstadt, 64283 Darmstadt, Germany (e-mail: robin.neuder@tu-darmstadt.de;
 alejandro.jimenez\_saez@tu-darmstadt.de).}
 }
 }

\usepackage{amsfonts,amsmath,amsthm,bm}
\usepackage{cite}
\usepackage{amsmath,amssymb,amsfonts}
\usepackage{algorithmic}
\usepackage{graphicx}
\usepackage{textcomp}
\usepackage{xcolor}
\usepackage[T1]{fontenc}
\usepackage{comment}
\usepackage{eucal}
\usepackage{algorithm}
\usepackage{algorithmic}
\usepackage[acronyms,nonumberlist,nopostdot,nomain,nogroupskip]{glossaries}
\usepackage[font=footnotesize]{caption}
\usepackage{subcaption}
\usepackage{booktabs,flushend}

\usepackage{hyperref}
\hypersetup{
    colorlinks=true,
    linkcolor=blue,
    citecolor=blue,
    urlcolor=black
}
\usepackage{orcidlink}
\usepackage{microtype}

\newtheorem{lem}{Lemma}

\newtheorem{remk}{Remark}

\newcommand{\defeq}{\triangleq}
\def\bOmega{\boldsymbol{\Omega}}

\newcommand{\e}{\mathsf{e}}
\newcommand{\jj}{\mathsf{j}}
\newcommand{\Herm}{\mathsf{H}}
\newcommand{\Trans}{\mathsf{T}}
\newcommand{\x}{\mathsf{x}}
\newcommand{\y}{\mathsf{y}}
\newcommand{\z}{\mathsf{z}}
\newcommand{\bA}{\mathbf{A}}
\newcommand{\bx}{\mathbf{x}}

\newcommand{\bs}{\mathbf{s}}
\newcommand{\bS}{\mathbf{S}}

\newcommand{\bH}{\mathbf{H}}
\newcommand{\ba}{\mathbf{a}}

\newcommand{\bC}{\mathbf{C}}
\newcommand{\bu}{\mathbf{u}}

\newcommand{\bI}{\mathbf{I}}

\newcommand{\bh}{\mathbf{h}}

\newcommand{\bq}{\mathbf{q}}
\newcommand{\bp}{\mathbf{p}}

\newcommand{\kk}{\kappa}

\newcommand{\bSigma}{\boldsymbol{\Sigma}}
\newcommand{\bGamma}{\boldsymbol{\Gamma}}
\newcommand{\bmu}{\boldsymbol{\mu}}

\newcommand{\blambda}{\boldsymbol{\lambda}}

\newcommand{\bbeta}{\boldsymbol{\beta}}
\newcommand{\bPhi}{\boldsymbol{\Phi}}

\newcommand{\Ex}{\mathbb{E}}

\newcommand{\diag}{\mathrm{diag}}
\newcommand{\real}{\mathrm{Re}}

\newcommand{\tr}{\mathrm{tr}}
\newcommand{\rank}{\mathrm{rank}}

\newcommand{\tx}{\mathrm{tx}}

\newcommand{\SNR}{\mathrm{SNR}}
\newcommand{\RS}{\mathrm{SR}}
\newcommand{\RIS}{\mathrm{RIS}}

\def\bomega{\boldsymbol{\omega}}

\def\bzero{\boldsymbol{0}}
\def\bone{\boldsymbol{1}}
\def\Cset{\mathbb{C}}

\def\Rset{\mathbb{R}}

\def\LOS{\mathrm{LOS}}

\def\tmax{\mathrm{max}}
\def\tx{\mathrm{tx}}
\def\rx{\mathrm{rx}}
\def\BS{\mathrm{BS}}
\def\RIS{\mathrm{RIS}}

\def\SNR{\mathrm{SNR}}

\def\eff{\mathrm{eff}}

\def\sCN{\mathcal{CN}}

\def\Pset{\mathcal{P}}
\def\bigO{\mathcal{O}}
\usepackage{xcolor}

\newacronym{RIS}{RIS}{reconfigurable intelligent surface}
\newacronym{QoS}{QoS}{quality of service}
\newacronym{LC}{LC}{liquid crystal}
\newacronym{SNR}{SNR}{signal-to-noise ratio}
\newacronym{TDMA}{TDMA}{time-division multiple-access}
\newacronym{BS}{BS}{base station}
\newacronym{MU}{MU}{mobile user}
\newacronym{ME}{ME}{mobile eavesdropper}
\newacronym{NF}{NF}{near-field}
\newacronym{Tx}{Tx}{transmitter}
\newacronym{Rx}{Rx}{receiver}
\newacronym{AWGN}{AWGN}{additive white Gaussian noise}
\newacronym{wrt}{w.r.t.}{with respect to}
\newacronym{RDE}{RDE}{Reaction-Diffusion Equation}
\newacronym{PDE}{PDE}{partial differential equation}
\newacronym{UPA}{UPA}{uniform planar array}
\newacronym{ULA}{ULA}{uniform linear array}
\newacronym{AO}{AO}{alternative optimization}
\newacronym{SOCP}{SOCP}{second-order cone programming}
\newacronym{AoD}{AoD}{angle of departure}
\newacronym{LOS}{LOS}{line of sight}
\newacronym{nLOS}{NLOS}{non-LOS}
\newacronym{MIMO}{MIMO}{multiple-input multiple-output}
\newacronym{RS}{SR}{secure rate}
\newacronym{SDP}{SDP}{semi-definite programming}
\newacronym{6G}{6G}{sixth generation}
\newacronym{CSI}{CSI}{channel state information}
\newacronym{PIN}{PIN}{positive-intrinsic-negative}
\newacronym{RF}{RF}{radio frequency}
\newacronym{MEMS}{MEMS}{micro-electro-mechanical system}
\newacronym{mmWave}{mmWave}{millimeter wave}
\newacronym{OFDM}{OFDM}{orthogonal frequency division multiplexing}
\newacronym{LSE}{LSE}{log-sum-exp}

\begin{document}

\maketitle

\begin{abstract}
Despite their advantages in scalability and energy efficiency for realizing extremely large \glspl{RIS}, \gls{LC} elements exhibit frequency-dependent phase-shift characteristics. Ignoring this inherent property can severely impact the overall system performance, particularly in wideband scenarios. This challenge is particularly problematic for physical layer security, as even slight frequency-induced phase deviations can cause significant unintended signal leakage.
%\Gls{LC} is a promising hardware solution for implementing large \glspl{RIS}, as it is cost-effective, energy-efficient, scalable, and capable of providing continuous phase shifts with low power consumption. However, the phase-shift response of \gls{LC}-based \glspl{RIS} is inherently frequency-dependent. If unaddressed, this characteristic leads to performance degradation, particularly in wideband scenarios. This issue is especially critical in secure communication applications, where minor phase-shift variations across elements can result in considerable information leakage.
This paper addresses these frequency-induced variations by developing a physics-based model for an \gls{LC} unit cell across varying frequencies and proposing a novel phase-shift design framework that maximizes secure communication across all subcarriers. Furthermore, given the large number of elements in \gls{mmWave} \gls{LC}-\glspl{RIS}, acquiring full \gls{CSI} is often impractical. Therefore, we optimize the phase shifts based solely on the locations of the legitimate \glspl{MU} and potential eavesdroppers. 
Additionally, we optimize the \gls{RIS} to illuminate a specific spatial zone instead of focusing on a precise target point.
This area illumination strategy improves robustness against inaccurate positioning data for both the \glspl{MU} and eavesdroppers, and reduces the overhead associated with frequent re-acquisition of positioning data.
%Moreover, rather than targeting a single user point, the \gls{RIS} is designed to illuminate a broader area. This approach enhances communication reliability for the \glspl{MU} and mitigates performance degradation caused by location estimation errors for both the \glspl{MU} and eavesdroppers.
To solve the formulated problem, we introduce both a \gls{SDP}-based solution and a low-complexity heuristic method. While the \gls{SDP}-based approach yields superior performance, it incurs higher computational complexity. Conversely, the scalable method exhibits a much slower scaling of complexity with respect to the number of \gls{RIS} elements, which makes it highly suitable for extremely large \glspl{RIS}. Our simulation results demonstrate that both proposed algorithms improve the secrecy rate compared to baseline methods that neglect frequency-dependent effects. In our considered setup operating at 60 GHz, the proposed \gls{SDP}-based and low-complexity frameworks achieve a secrecy rate of approximately 2 and 1 bits/symbol, respectively, while benchmark schemes cannot guarantee a non-zero secrecy rate across the entire 8 GHz bandwidth. Finally, the proposed wideband phase-shift design is validated through experimental evaluations on an \gls{LC}-\gls{RIS} setup, which has confirmed its practical effectiveness.

\begin{IEEEkeywords}
Liquid Crystal, reconfigurable intelligent surface, wideband communications, secure communications, frequency-dependency.
 \end{IEEEkeywords}
\end{abstract}
\glsresetall
\section{Introduction}
%\glsentryshort{RIS}
\Glspl{RIS} have emerged as a promising technology for next-generation wireless networks by enabling highly programmable propagation environments \cite{Wu2019,Basar2019,di2019smart,najafi2020physics}. Recently, \gls{LC} technology has gained attraction as a scalable and energy-efficient solution for deploying large-scale \glspl{RIS}, particularly at \gls{mmWave} frequencies \cite{neuder2024architecture,aboagye2022design,delbari2026fast}. The literature has extensively investigated \gls{LC}-based \glspl{RIS} from both theoretical and practical perspectives. For example, \cite{neuder2024architecture} demonstrated a compact hardware prototype, while \cite{aboagye2022design} evaluated its potential in visible light communications. System-level metrics, including cost and energy efficiency, were benchmarked against competing technologies in \cite{jimenez2023reconfigurable}. Furthermore, recent works have modeled the temporal transition behaviors of \glspl{LC} in time-division multiple-access setups \cite{delbari2024fast,delbari2026fast} and analyzed their thermal dependencies \cite{delbari2024temperature,gholian2025temperature}. In addition, \cite{Li2024,Wang2025} discussed and analyzed the loss model of \gls{LC}-\gls{RIS}. Together, these studies highlight a growing momentum toward refining the practical viability of \gls{LC}-driven \glspl{RIS}.

The phase-shift response of \gls{LC}-based \glspl{RIS} is fundamentally frequency-dependent \cite{neuder2024architecture}, which can cause severe performance bottlenecks in wideband deployments. Existing literature addressing wideband challenges includes near-field beamforming \cite{Cui2024,Yu2024}, advanced signal processing \cite{Qian2025}, and novel system architectures \cite{He2021,Su2023,Mo2024} to mitigate beam splitting, which all rely predominantly on idealized \gls{RIS} models. These studies typically assume that phase shifts are either frequency-independent \cite{Yu2024,Qian2025,He2021,Su2023,Mo2024} or perfectly adjustable per subcarrier \cite{Cui2024}. Furthermore, even approaches utilizing more realistic hardware models \cite{li2021intelligent} generally depend on the restrictive assumption of perfect user location data for beamforming.

The performance degradation due to the frequency-dependency is particularly detrimental in physical layer security, where \gls{OFDM} transmissions could suffer from unintended information leakage. Physical layer secrecy in \gls{RIS}-assisted systems has been explored in several works in the literature \cite{Shen2019,Asaad2022,Chu2021,Xiu2021,huang2025secrecy}. For instance, studies have investigated the maximization of weighted secrecy sum-rates in generic \gls{MIMO} wiretap settings using fractional programming and alternating optimization \cite{Asaad2022}, as well as leveraging artificial noise alongside \gls{RIS} phase shifts to further degrade eavesdropper reception \cite{Chu2021}. Hardware-constrained scenarios have also been addressed, such as \gls{mmWave} systems employing low-resolution digital-to-analog converters, where algorithms based on successive convex approximation and block coordinate descent are used to mitigate hardware losses and optimize phase shifts \cite{Xiu2021}. Furthermore, robust designs have been developed to tackle the practical challenges of imperfect \gls{CSI} in multi-user and multi-eavesdropper settings, utilizing methods like \gls{SDP} combined with iterative hybrid optimization algorithms to maintain high confidentiality rates despite channel estimation errors \cite{huang2025secrecy}. 

In contrast, this paper tackles the phase-shift design problem under the practical hardware constraints of \gls{LC}-\glspl{RIS}, operating without precise knowledge of legitimate \glspl{MU} and eavesdropper locations. Specifically, we propose a secure illumination strategy that maximizes signal strength within a target zone occupied by legitimate \glspl{MU}, while actively suppressing power in surrounding areas vulnerable to eavesdropping. By explicitly incorporating the frequency-dependent response of individual \gls{LC}-\gls{RIS} elements into a novel optimization framework, we bridge the gap between idealized wideband beamforming theory and practical, secure spatial coverage\footnote{Although the proposed model and design are simulated within a secrecy context, they can also be applied to other scenarios, e.g., serving a single user (see Section~\ref{sec: Experimental Results}).}. To the best of the authors' knowledge, this is the first study to investigate the frequency-dependent characteristics of \gls{LC}-\glspl{RIS} within a wideband context. Our main contributions are summarized as follows:

\begin{itemize}
\item \textbf{Physics-based Modeling of \gls{LC}-\gls{RIS} Phase-shift over Different Frequencies:} First, we develop a rigorous physical model characterizing how frequency variations dictate the phase-shift response of \gls{LC}-\gls{RIS} elements. This model characterizes the phase shift induced at a given frequency as a function of the reference phase shift e.g., at the central carrier frequency.
\item \textbf{Problem Formulation:} Next, we introduce a realistic physical layer security framework including multiple legitimate users and a mobile eavesdropper. We assume the exact locations and instantaneous \gls{CSI} of all parties are unknown, relying solely on approximate user vicinity data. This formulation reflects the practical challenges of securing \gls{RIS}-aided communications in real-world environments where users are movable and the location estimation errors are unavoidable. Then, we formulate a non-convex optimization problem to jointly design the \gls{LC}-\gls{RIS} phase shifts while accounting for their hardware-specific frequency dependencies.
\item \textbf{Proposed Algorithms:} To address the inherent non-convexity, we apply algebraic reformulations and propose a highly efficient, sub-optimal solution leveraging rank relaxation and \gls{SDP}. Although the \gls{SDP}-based method achieves excellent performance, its computational complexity scales cubically with the number of \gls{RIS} elements. Alternatively, in addition to the conference version \cite{delbari2025wideband}, we propose a low-complexity method that exhibits significantly lower computational complexity, making it highly suitable for extremely large \glspl{RIS}.
\item \textbf{Performance Evaluation:} Subsequently, we validate the proposed algorithms through extensive simulations in addition to the conference version \cite{delbari2025wideband}. First, we evaluate the convergence behavior and empirically analyze the computational complexity of the two methods by comparing their execution times. Next, we assess their performance against three benchmark schemes that ignore frequency-dependent phase responses. The results demonstrate that explicitly accounting for these characteristics in the \gls{LC}-\gls{RIS} design significantly enhances the guaranteed secrecy rate across all \gls{OFDM} subcarriers. Furthermore, we show that while the secrecy rate achieved by the \gls{SDP}-based algorithm surpasses that of the scalable approach, this performance gain comes at the cost of significantly higher computational complexity.
\item \textbf{Experimental Verification:} Finally, unlike the conference version \cite{delbari2025wideband}, we validate our approach through real-world experiments in an indoor scenario. We observe that the optimized phase-shift design indeed enhances the received power at the \gls{MU} location compared to a benchmark scheme that neglects the impact of frequency dependency inherent to the \gls{RIS} elements.
\end{itemize}

The remainder of this paper is organized as follows. In Section \ref{sec: System and Channel Models}, we present the system, channel, and frequency response of \gls{LC} models. Section \ref{sec: OFDM LC-RIS Phase-shift Design} details the proposed \gls{SDP}-based optimization and low-complexity method frameworks, followed by simulation results in Section \ref{sec: Performance Comparison}. Then, we test our algorithm in an experimental implementation in Section~\ref{sec: Experimental Results}. Finally, Section \ref{sec: Conclusion} concludes the paper.

\begin{figure}
    \centering
    \includegraphics[width=0.4\textwidth]{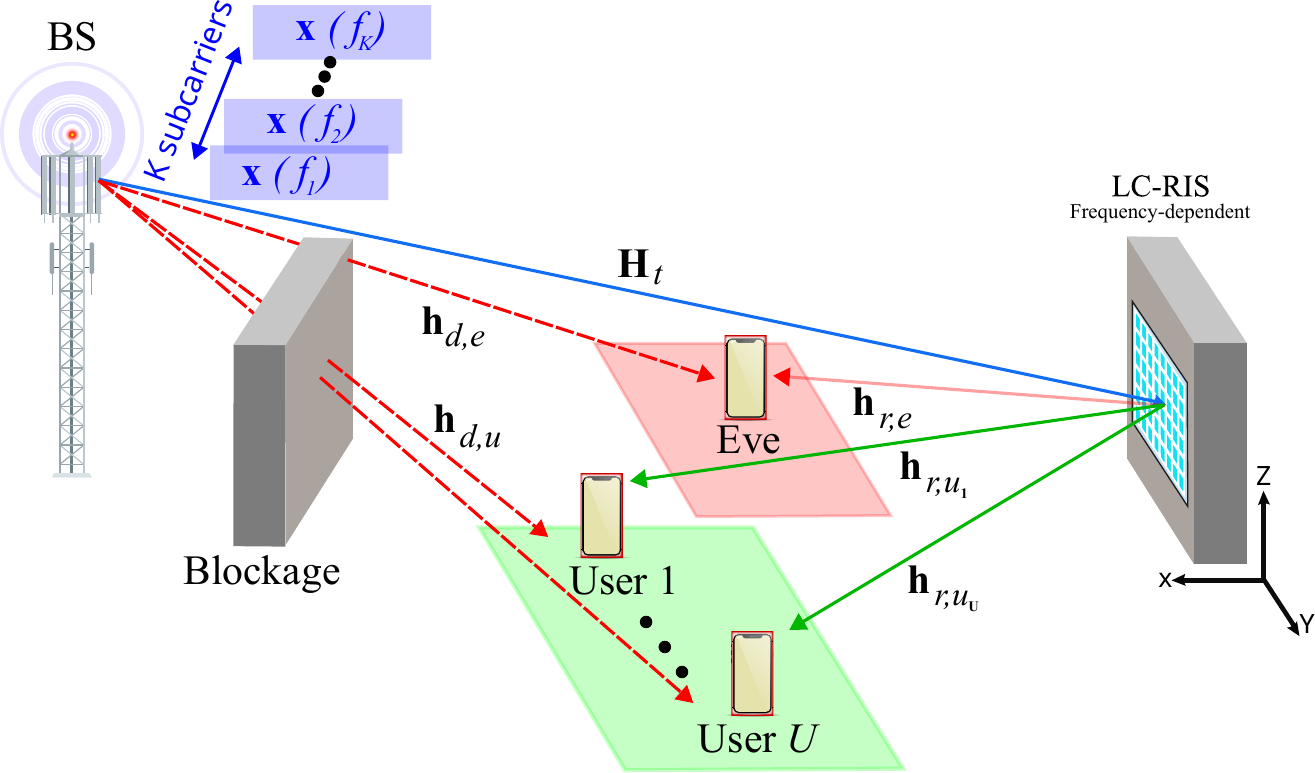}
    \caption{An illustration of a wireless communication setup in which a \gls{BS} communicates with authorized users through an \gls{RIS}, while simultaneously minimizing the signal intercepted by an eavesdropper.}
    \label{fig:system model}
    \vspace{-5 mm}
\end{figure}

\textit{Notation:} Bold capital and small letters are used to denote matrices and vectors, respectively.  $(\cdot)^\Trans$, $(\cdot)^\Herm$, $\rank(\cdot)$, $\tr(\cdot)$, $\det(\cdot)$, $\odot$, and $(\cdot)^{\circ}$ denote the transpose, Hermitian, rank, trace, determinant of a matrix, Hadamard product, and Hadamard power, respectively. Moreover, $\diag(\bA)$ is a vector that contains the main diagonal entries of matrix $\bA$, $\bone_n$ and $\bzero_n$ denote column vectors of size $n$ whose elements are all ones and zeros, respectively. $\|\bA\|_*=\sum_i \sigma_i$, $\|\bA\|_2=\max_i \sigma_i$, $\|\bA\|_F$, and $\blambda_{\max}(\bA)$ denote the respectively nuclear, spectral, and Frobenius norms of a Hermitian matrix $\bA$, and eigenvector associated with the maximum eigenvalue of matrix $\bA$, where $\sigma_i,\,\,\forall i$, are the singular values of $\bA$. Furthermore, $[\bA]_{m,n}$ and $[\ba]_{n}$ denote the element in the $m$th row and $n$th column of matrix $\bA$ and the $n$th entry of vector $\ba$, respectively. $x^+$ denotes as $\max\{x,0\}$. Moreover, $\Rset$ and $\Cset$ represent the sets of real and complex numbers, respectively, $\jj$ is the imaginary unit, and $\Ex\{\cdot\}$  represents expectation. $\mathrm{rand}(N)$ denotes a $N\times1$ vector where each element is generated independently and uniformly from 0 to 1. $\mathcal{CN}(\bmu,\bSigma)$ denotes a complex Gaussian random vector with mean vector $\bmu$ and covariance matrix $\bSigma$. Finally, $\bigO(\cdot)$ represents the big-O notation and $|\Pset|$ is the cardinality of set $\Pset$.

\section{System, Channel, and LC Models}
\label{sec: System and Channel Models}
In this section, we firstly present the system model, including $U$ legitimate users and a mobile eavesdropper, followed by the channel model. Then, we discuss the \gls{LC}-\gls{RIS} phase-shift model principle. Finally, we define the secrecy rate used throughout this paper.

\subsection{System Model}
\label{sec: system model}
In this paper, we study a wideband downlink communication scenario involving a \gls{BS} equipped with $N_t$ transmitter antennas, an \gls{LC}-\gls{RIS} with $N$ \gls{LC}-based unit cells, a single-antenna eavesdropper, and $U$ single-antenna legitimate users. The received signal frequency representations at \gls{ME} ($y_e(f_k)\in\Cset$) and $u$th legitimate \gls{MU} ($y_u(f_k)\in\Cset$) are expressed as
\begin{align}
\label{Eq:system model user}
	y_g(f_k) = &\big(\bh_{d,g}^\Herm(f_k) + \bh_{r,g}^\Herm(f_k) \bGamma(f_k) \bH_t(f_k) \big) \bx(f_k)\! +\!n_g,\\
    &\quad\,\, g=\{u_1,\cdots,u_U,e\}, \quad f_k\in\{f_1, \cdots, f_K\}.\nonumber
\end{align}
Let $f_k = f_1, \cdots, f_K$ denote the subcarriers in the \gls{OFDM} setup, where $K$ is the total number of subcarriers. $\bx(f_k)\in\Cset^{N_t}$ represents the transmit signal vector for the legitimate users on the $f_k$th subcarrier, and  $n_u\in\Cset$ ($n_e\in\Cset$) is the \gls{AWGN} at the legitimate user (eavesdropper), i.e., $n_u,n_e\sim\sCN(0,\sigma_n^2)$, with $\sigma_n^2$ being the noise power. We adopt hybrid beamforming, where for each subcarrier, the transmit signal vector $\bx(f_k)$ is given by $\bx(f_k) = \bq s(f_k)$ where $\bq\in \Cset^{N_t}$ is beamforming vector on the \gls{BS}, and $s(f_k)\in \Cset^{N_s}$ is the transmitted data symbol at subcarrier $k$ with $\Ex\{|s(f_k)|^2\}=1,\,\forall k$. The beamformer must satisfy the transmit power constraint $\|\bq\|^2 \leq P_t$, where $P_t$ denotes the maximum allowable transmit power. The \gls{RIS} reflection is characterized by a diagonal matrix $\bGamma(f_k) \in \Cset^{N \times N}$, whose $n$th diagonal entry is given by $[\bGamma]_n(f_k) = [\bOmega]_n \e^{\jj[\bomega]_n(f_k)}$, where $[\bomega]_n(f_k)$ and $[\bOmega]_n$ denote the phase shift and reflection amplitude of the $n$th unit cell, respectively. Based on the design considerations of \glspl{LC}-\glspl{RIS} in the relevant frequency band \cite{neuder2024architecture}, we assume negligible amplitude variation across unit cells, i.e., $[\bOmega]_n \approx 1,\ \forall n$ in all subcarriers. Moreover,  $\bh_{d,g}(f)\in\Cset^{N_t}, \bH_t(f)\in\Cset^{N\times N_t}$, and $\bh_{r,g}(f)\in\Cset^{N}$ denote the \gls{BS}-\{MU, ME\}, BS-RIS, and RIS-\{MU, ME\} frequency selective channel matrices, respectively, where $g\in\{u_1,\cdots,u_U,e\}$, and $f$ is frequency.

\subsection{Channel Model}
\label{sec: channel model}
Because \gls{mmWave} \gls{RIS} deployments are often elevated to resolve the impact of the blockages, the resulting channels are almost dominated by \gls{LOS} propagation rather than \gls{nLOS} links. Additionally, the large aperture size of \gls{LC}-\gls{RIS} arrays dictates that \glspl{MU} and \gls{BS} frequently reside in the radiating \gls{NF} region, which necessitates the use of a generalized \gls{NF} \gls{MIMO} Rician channel model proposed in \cite{delbari2025near} for accurate system characterization \cite{Liu2023nearfield,delbari2024nearfield}. To simplify the notation, we discuss the general model $\bH\in\Cset^{N_\rx\times N_\tx}$ where $N_\rx$ and $N_\tx$ are the number of \gls{Rx} and \gls{Tx} antennas, respectively.
%Given the assumption of an extremely large \gls{LC}-\gls{RIS}, the distances between the RIS and both the \gls{BS} and the \gls{MU} may fall within the \gls{NF} regime of the RIS \cite{Liu2023nearfield,delbari2024nearfield}. Therefore, an \gls{NF} channel model is adopted. Furthermore, RISs are typically installed at elevated heights, ensuring \gls{LOS} links between the RIS and both the BS and \gls{MU}s. High-frequency further emphasizes the dominance of these LOS links over \gls{nLOS} links. As a result, the channels are modeled using the generalized \gls{NF} \gls{MIMO} Rician model proposed in \cite{delbari2025near}. 
This \gls{MIMO} channel model can be written as
\begin{equation}
\label{eq: channel model}
\bH=c_0(\bH^{\mathrm{LOS}}+\sum_{r=1}^R\bar{k}_r\bar{\bH}_r+\tilde{k}_r\tilde{\bH}_r),
\end{equation}
where $\bH^\LOS$ denotes the LOS \gls{NF} channel matrix, and $\bar{\bH}_r$ is the deterministic component of the \gls{nLOS} channel resulting from reflections off large objects in $r$th path, such as the ground. Furthermore, $\tilde{\bH}_r$ represents the stochastic component of the \gls{nLOS} channel for the $r$-th reflector. The variables $c_0$, $\bar{k}_r$, and $\tilde{k}_r$ denote the channel amplitude of the \gls{LOS} path and the $K$-factors for the deterministic and stochastic parts of the $r$-th nLOS channel, respectively \cite{delbari2025near}. In \gls{NF} regime, $\bH^\LOS$ and $\bar{\bH}_r,\,\forall r$ are functions of the locations of antennas in order to capture the underlying spherical wave propagation. This leads to \cite{Liu2023nearfield}
\begin{subequations}
\label{eq: channel models}
\begin{align}
	[\bH^\LOS]_{m,n} &= \, \e^{\jj\kk\|\bu_{\rx,m}-\bu_{\tx,n}\|}\label{Eq:LoSnear},\\
    [\bar{\bH}_r]_{m,n} &= \, \e^{\jj\kk\|\bu_{\rx,m}^r-\bu_{\tx,n}\|},\\
    [\tilde{\bH}_r]_{m,n} &\sim\sCN(0,1),
\end{align}
\end{subequations}
where $\bu_{\rx,m}$ and $\bu_{\tx,n}$ are the locations of the $m$th \gls{Rx} antenna and the $n$th \gls{Tx} antenna, respectively. Moreover, $\bu_{\rx,m}^r$ is the image of the location $\bu_{\rx,m}$ reflected by the $r$th reflector. In addition, $\kk=2\pi f/c$ is the wave number with $c$ being the speed of light in vacuum.  This model is applicable to all channel components $\bH_t$, $\bh_{r,g}$, and $\bh_{d,g}$, for $g\in\{u_1,\cdots,u_U,e\}$, with appropriate adaptations.

We include an extra barrier penetration in \gls{BS}-\glspl{MU} \gls{LOS} link to model the existence of a blockage based on the model suggested in \cite{Phillips2013}. Therefore, we assume $\bh_{d,g} \approx \bzero_{N_t},\,\forall g\in\{u_1,\cdots,u_U,e\},$ in optimization part (Section~\ref{sec: OFDM LC-RIS Phase-shift Design}). Nevertheless, we consider the impact of this blockage in the simulation results (Section~\ref{sec: Performance Comparison}).
%For the \gls{LOS} link between the \gls{BS} and the \gls{MU}s, we incorporate an additional high penetration loss due to a blockage, as suggested in \cite{Phillips2013}. Consequently, in the optimization section, we assume $\bh_{d,g} \approx \bzero_{N_t},\,\forall g\in\{u_1,\cdots,u_U,e\}$ while we account for its impact in the simulation results in Section~\ref{sec: Performance Comparison}.

\subsection{Frequency Response of an LC Cell}
\label{Frequency response of the LC cell}
This subsection details the physical mechanism by which \gls{LC} molecules impart a phase shift onto incident electromagnetic waves, as well as the frequency-dependent nature of this response. \gls{LC}-based \glspl{RIS} control signal reflections by leveraging the tunable electromagnetic characteristics of \gls{LC} molecules, whose spatial orientation can be dynamically altered via an external biasing voltage \cite{jimenez2023reconfigurable}. This voltage-driven reorientation directly changes the effective permittivity of the \gls{LC} medium, which in turn dictates the phase shift generated by each individual \gls{RIS} element.

Due to the elongated, rod-like structure of \gls{LC} molecules, their electromagnetic response is highly anisotropic, depending heavily on the alignment between the incident \gls{RF} electric field ($\vec{E}_{\rm RF}$) and the molecules' principal axes. Specifically, when ($\vec{E}_{\rm RF}$) aligns with the major axis of the molecules, the effective permittivity increases, resulting in a larger induced phase shift, as depicted in Fig.~\ref{fig:V_phase}. Consequently, continuously tuning the applied voltage allows for precise control over the molecular orientation. This capability empowers the \gls{RIS} to dynamically shape signal reflections, thereby facilitating a fully programmable wireless propagation environment.
\begin{figure}[tbp]
	\centering
    \includegraphics[width=0.4\textwidth]{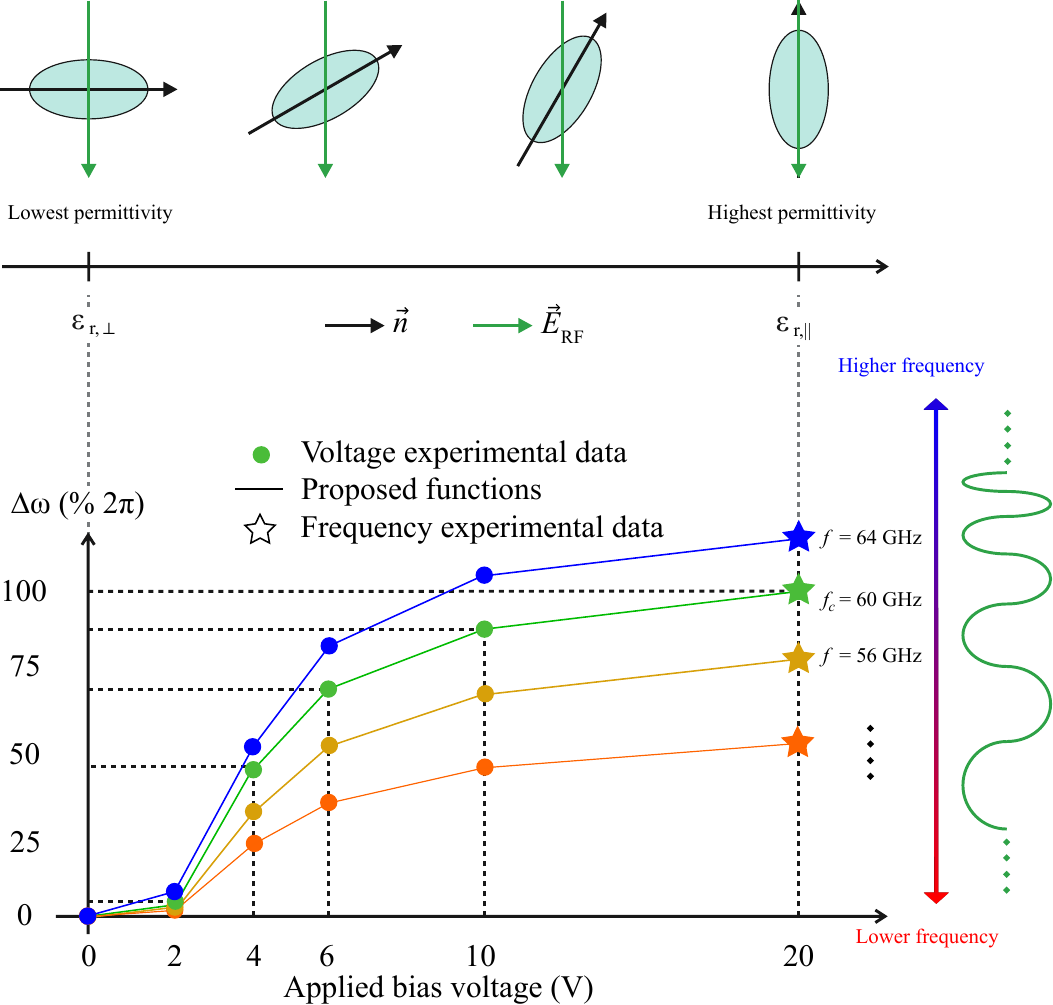}
    \caption{The relationship between phase shift and applied voltage for an \gls{LC} for different frequencies is depicted. The experimental data is given by \cite[Fig.~3.e]{neuder2024architecture}, and a piece-wise linear function is modeled similarly to \cite{delbari2024fast}.}
    \vspace{-5mm}
    \label{fig:V_phase}
\end{figure}

The \gls{LC} unit maximum achievable phase-shift range is determined by several parameters, including the length of the phase shifter, the permittivity differential corresponding to distinct molecular orientations, and the operating frequency. This relationship is expressed as \cite{garbovskiy2012liquid}:
\begin{equation}
    \Delta\omega_{\max}(f)=2\pi l\Delta n_{\max}\frac{1}{c}\left(f_c+\beta\left(f-f_c\right)\right),
    \label{eq:omega epsilon}
\end{equation}
where $l$, $f$, $f_c$, and $c$ denote the phase shifter length, the operating frequency, the center frequency, and the speed of light in vacuum, respectively. The factor $\beta$ controls how the maximum tuning range of the phase shifter varies with frequency. The term $\Delta n_{\max}$ represents the maximum possible birefringence of the \gls{LC} and is defined as
\begin{equation}
    \Delta n_{\max}=\sqrt{\varepsilon_{r,\parallel}}-\sqrt{\varepsilon_{r,\perp}},
    \label{eq:delta n}
\end{equation}
where $\sqrt{\varepsilon_{r,\parallel}}$ and $\sqrt{\varepsilon_{r,\perp}}$ are the maximum and the minimum relative permittivity, respectively. These extreme values are observed when the applied electric field is parallel and perpendicular, respectively, to the director vector $\vec{n}$, as depicted in Fig.~\ref{fig:V_phase}. Adjusting the input voltage modulates the effective birefringence within the range $0 \leq \Delta n \leq \Delta n_{\max}$, thereby enabling precise control over the phase shift at any specific frequency. Assuming that the maximum differential phase shift at the center frequency ($f_c$) is $2\pi$, we can establish a normalized control framework. By setting the minimum reference phase shift to zero (i.e., $\omega_{\min}=0$) and defining $h(v) \defeq 2\pi\frac{\Delta n}{\Delta n_{\max}}$, where $0 \leq h(v) \leq 2\pi, \quad \forall v$, characterizes the empirical voltage-to-phase mapping at a given frequency, we can reformulate \eqref{eq:omega epsilon} as follows:
\begin{equation}
\omega(f)\!=\!2\pi\frac{2\pi l\Delta n\frac{1}{c}\left(f_c+\beta\left(f-f_c\right)\right)}{2\pi l\Delta n_{\max}\frac{f_c}{c}}\!=\!h(v)\!\!\left(\!1\!+\!\beta(\frac{f}{f_c}-1)\!\!\right).
\label{eq:omega reference f}
\end{equation}
Although a closed-form analytical expression for $h(v)$ is unavailable, it can be determined empirically, as illustrated by the green curve in Fig.~\ref{fig:V_phase}. Furthermore, the frequency dependence of the phase shift $\omega$ aligns with the experimental observations reported in \cite[Fig.~3e]{neuder2024architecture}. Specifically, \eqref{eq:omega reference f} demonstrates that the maximum achievable phase shift scales proportionally with frequency variations, a behavior corroborated by both Fig.~\ref{fig:V_phase} and the empirical results in \cite{neuder2024architecture}.

Instead of optimizing the control voltages directly, we define a reference phase shift at the center frequency as $\omega_c\defeq\omega(f_c)$. In Section~\ref{sec: OFDM LC-RIS Phase-shift Design}, the optimization is performed over the variables $[\bomega_c]_n,\,\forall n$. For the remaining subcarriers, the corresponding phase shifts are intrinsically coupled through the following constraint:
\begin{equation}
    \label{eq: omega to omega_c}
    [\bomega(f_k)]_n=[\bomega_c]_n\!\!\left(\!1\!+\!\beta(\frac{f_k}{f_c}-1)\!\!\right)=[\bomega_c]_n\beta_k,\forall n,\forall k,
\end{equation}
where we defined $\beta_k\defeq\left(\!1\!+\!\beta(\frac{f_k}{f_c}-1)\!\right)$ and \eqref{eq: omega to omega_c} ensures that the phase shifts scale consistently across all subcarriers.

\subsection{Secrecy Rate}
\label{sec: Secrecy rate}
The secrecy rate serves as a fundamental metric for evaluating the performance of physical-layer security in communication systems. It is formally defined as the difference between the achievable data rate of the legitimate user and that of the eavesdropper 
\begin{subequations}
    \label{eq: SNR}
\begin{align}
    \RS(f_k)&=[\widetilde{\RS}(f_k)]^+,\\
    \widetilde{\RS}(f_k) &=\log(1\!+\!\SNR_u(f_k))\!-\!\log(1\!+\!\SNR_e(f_k)),\\
    \SNR_u(f_k)&=|(\bh_u^\eff(f_k))^\Herm\bq|^2/\sigma^2_n,\\
    \SNR_e(f_k)&=|(\bh_e^\eff(f_k))^\Herm\bq|^2/\sigma^2_n,
\end{align}
\end{subequations}
where $(\bh_g^\eff)^\Herm=\bh_{d,g}^\Herm + \bh_{r,g}^\Herm \bGamma \bH_t$ for $g \in \{u_1,\cdots,u_U,e\}$ represents the effective end-to-end channel extending from the \gls{BS} to both the legitimate users and the eavesdropper.

In this study, we aim to maximize the secrecy rate ($\RS(f_k)$) across all frequency subcarriers of interest by jointly optimizing the \gls{RIS} phase shifts and the \gls{BS} beamformer. This joint optimization enables the system to construct a favorable propagation environment for legitimate users while simultaneously degrading the signal quality within the eavesdropper's vicinity\footnote{In this work, we do not address the \gls{OFDM} subcarrier allocation process. Instead, our focus is strictly on optimizing the phase shifts to guarantee robust link quality across the entire wideband spectrum (i.e., resolving the secure illumination problem). Given an \gls{RIS} configuration, orthogonal subcarrier assignment can be applied separately on a smaller time scale and hence falls outside the scope of this paper.}. We base the \gls{RIS} phase-shift design on the approximate geographical regions of the target coverage area and the eavesdropper, which offers the following practical advantages:
\begin{itemize}
    \item Diverging from the conventional assumption in the literature that full (instantaneous or statistical) \gls{CSI} of the eavesdropper is available, we only assume that the eavesdropper's location is merely confined to a specified spatial region, $\bp_e\in\Pset_e$. The physical extent of this area, $|\Pset_e|$, intrinsically compensates for potential location estimation uncertainties.
    \item Similarly, we assume the legitimate users reside somewhere within a predefined target zone, $\bp_u\in\Pset_u$. Expanding the size of this area, $|\Pset_u|$, not only mitigates localization errors but also significantly lowers the control overhead by reducing the need for frequent \gls{RIS} reconfigurations \cite{Jamali2022lowtozero}. Consequently, the \gls{RIS} phase shifts are optimized to simultaneously serve legitimate users across all potential locations and subcarriers.
\end{itemize}
Finally, the \gls{RIS} phase-shift design is predominantly driven by the \gls{LOS} paths. At \gls{mmWave} frequencies, the primary deployment spectrum for \gls{LC}-\glspl{RIS}, the \gls{LOS} component dictates the dominant contribution to the overall received signal strength, making it the most critical factor for optimizing secure communication performance. The impact of the \gls{nLOS} paths will be investigated in Section~\ref{sec: Performance Comparison}.

\section{OFDM LC-RIS Phase-shift Design}
\label{sec: OFDM LC-RIS Phase-shift Design}
In this section, we begin by formulating an optimization problem for \gls{RIS} phase-shift design aimed at maximizing the secure rate for all the subcarriers in \gls{OFDM} setup. Considering the setup depicted in Fig. \ref{fig:system model}, the \gls{RIS} aims to maximize the secure rate defined in \eqref{eq: SNR} for the \glspl{MU} under the worst-case scenario. Specifically, the objective is to maximize the secure rate regardless of the exact locations of all the legitimate users and the eavesdropper, as long as they remain within their respective designated areas. This requirement must be satisfied for all subcarriers $f_k$, where $k = 1, \cdots, K$. We formulate the following problem formulation based on the only \gls{LOS} links:
\begin{subequations}
\label{eq:optimization 1}
\begin{align}
    \text {P1:}\quad&~\underset{\alpha,\bomega_c,\bq}{\max}~\alpha \label{eq:optimization 1 a}
    \\&~\text {s.t.} ~~\RS(f_k)\geq \alpha,\, \forall \bp_u\in\Pset_u,\,\forall\bp_e\in\Pset_e,\,\forall k,\label{eq:optimization 1 b}
    \\&\quad\hphantom {\text {s.t.} } 0\leq [\bomega_c]_n < 2\pi, \label{eq:optimization 1 c}
    \\&\quad\hphantom {\text {s.t.} }  [\bomega(f_k)]_n = [\bomega_c]_n\beta_k, \forall n,\,\forall k, \label{eq:optimization 1 d}
    %\\&\quad\hphantom {\text {s.t.} } \sum_{k=1}^K\tr(\bF_k\bF_k^\Herm) \leq P_t.
    \\&\quad\hphantom {\text {s.t.} } \|\bq\|_2^2 \leq P_t. \label{eq:optimization 1 e}
\end{align}
\end{subequations}
In problem P1, (\ref{eq:optimization 1 a}) is the cost function, constraint (\ref{eq:optimization 1 b}) represents the secure rate condition defined in \eqref{eq: SNR}, (\ref{eq:optimization 1 c}) is the realizable phase-shift range for the center frequency, (\ref{eq:optimization 1 d}) enforces the phase shifter of each \gls{RIS} element across subcarriers as defined in \eqref{eq: omega to omega_c}, and (\ref{eq:optimization 1 e}) limits the transmit power. The parameter $\alpha$ characterizes the worst-case secure rate over all possible user and eavesdropper locations as well as all subcarriers, and the objective is to maximize this value. The optimization problem P1 is non-convex, primarily due to the non-convex nature of constraint (\ref{eq:optimization 1 b}). Additionally, the coupling between the variables $\bq$ and $\bomega(f_k),\,\forall k,$ in this constraint further complicates the derivation of a global solution. Before solving Problem P1, we show when beam squinting cannot be ignored as a function of \gls{Tx} (\gls{BS} or \gls{RIS}) size and bandwidth, motivated by the observations in Figs.~\ref{fig: f_N_fixed} and \ref{fig: N_f_fixed} in Section~\ref{Sec: Impact of Number of Antenna in Wideband Scenario}. Next, we decompose the original problem into two sub-problems and iteratively maximize the objective using the \gls{AO} method, as discussed in Section~\ref{sec: Beamformer Design} and Section~\ref{sec: RIS Phase Shifter Design}. Finally, we analyze the overall algorithmic complexity in Section~\ref{sec: Algorithm and Complexity Analysis}.
\subsection{Beam Squinting: The Interplay of Array Size and Bandwidth}
\label{Sec: Impact of Number of Antenna in Wideband Scenario}
Before solving Problem P1, we first study when beam squinting is actually relevant as a function of array size and bandwidth, and then use that insight to simplify the solution. Consider a scenario where an $N$-element \gls{ULA} is deployed at the \gls{Tx} located at the origin, communicating with a single-antenna \gls{Rx} positioned at a fixed distance. Let $\bh_k^\Herm$ represent the channel vector for the $k$-th subcarrier, and let $\bq$ denote the transmit beamforming vector. If the beamformer is designed exclusively for the center frequency $f_c$, it is expressed as $\bq_c = \frac{\bh_c}{|\bh_c|} \sqrt{P_t}$. Consequently, the received \gls{SNR} at subcarrier $k$ is given by $\SNR_k = \frac{|\bh_k^\Herm \bq_c|^2}{\sigma_n^2}$, with $\SNR_c$ denoting the \gls{SNR} at the center frequency, which represents the peak signal quality across the band. 

Fig.~\ref{fig: f_N_fixed} plots the \gls{SNR} of the $k$-th subcarrier normalized against the center frequency \gls{SNR}. Concurrently, Fig.~\ref{fig: N_f_fixed} depicts the worst-case (minimum) normalized \gls{SNR} evaluated across the entire frequency band. Based on these simulation results, we observe the following distinct trends:
\begin{itemize}
    \item For a fixed array size $N$, expanding the operational bandwidth causes the ratio $\frac{\SNR_k}{\SNR_c}$ to decline. The steepness of this degradation (i.e., the slope) is directly dictated by the number of antennas $N$, as illustrated in Fig.~\ref{fig: f_N_fixed}.
    \item Conversely, when the bandwidth is held constant and the number of \gls{Tx} antennas $N$ is increased, the minimum \gls{SNR} ratio across subcarriers, $\underset{k}{\min} \frac{\SNR_k}{\SNR_c}$, also decreases. As shown in Fig.~\ref{fig: N_f_fixed}, this performance penalty remains negligible for smaller values of $N$, but becomes severely detrimental once $N$ exceeds a specific threshold.
\end{itemize}

These observations show that approximating the wideband beamformer using the center frequency response may be valid for many \gls{MIMO} scenarios (\gls{BS} and \glspl{MU}) when the number of elements in one dimension of \gls{Tx} is not large. However, this approximation completely breaks down for extremely large \glspl{RIS} equipped with hundreds of elements\footnote{For a comprehensive mathematical proof of this phenomenon, we refer the reader to \cite{Parvini}.}. To simplify the proposed solution and focus on the large \gls{LC}-\gls{RIS} as the main contribution of this paper, we assume beam squinting is negligible for the \gls{MIMO} \gls{BS} but not for the \gls{LC}-\gls{RIS}. The proposed solution can be straightforwardly extended to massive \gls{MIMO} \gls{BS} as is already studied in \cite{Parvini}.

\begin{figure}[t]
\centering
\begin{subfigure}{0.4\textwidth}
\includegraphics[width=\textwidth,height=0.55\textwidth]{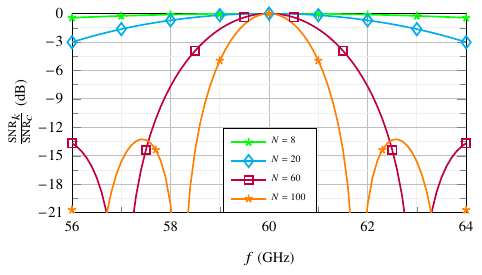}
    %\vspace{-0.5cm}
\caption{The normalized \gls{Rx} power (dB) versus frequency for given $N$ values 8, 20, 60, and 100.}
\label{fig: f_N_fixed}
\end{subfigure}
\begin{subfigure}{0.4\textwidth}
\includegraphics[width=\textwidth,height=0.55\textwidth]{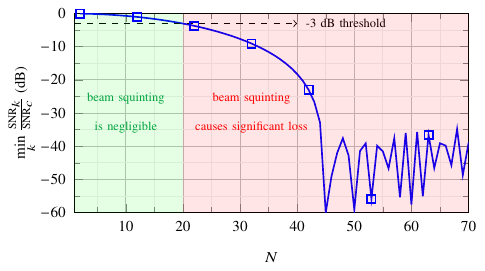}
    %\vspace{-0.5cm}
    \caption{The normalized minimum \gls{Rx} power~(dB) accross all subcarriers versus number of \gls{Tx} antenna for given bandwidth 8 GHz, where $f_c=60$~GHz.}
    \label{fig: N_f_fixed}
\end{subfigure}
\caption{Impact of the bandwidth and number of \gls{Tx} antenna on the \gls{Rx} power.}
\vspace{- 5 mm}
\label{fig: N and f different}
\end{figure}

\subsection{Beamformer Design}
\label{sec: Beamformer Design}
First, we fix $\bomega_c$ and treat $\bq$ as the sole optimization variable. Since our primary focus in this paper is on the \gls{RIS}, we assume that the number of elements in any single dimension of the \gls{BS} array falls within a regime where beam squinting is negligible\footnote{The impact of beam squinting at the \gls{BS} has been extensively investigated in the literature \cite{Parvini}, and the framework presented in this paper can be readily extended to massive \gls{MIMO} scenarios.} (see Fig.~\ref{fig: N_f_fixed}). Therefore, we should only optimize $\bq$ for the center frequency in this section. Furthermore, since the \gls{LOS} path is typically dominant at higher frequencies, we design the beamformer based solely on the \gls{LOS} component. This assumption is generally valid, particularly because both the \gls{BS} and \gls{RIS} are deployed at elevated positions above ground level. Under this assumption, the channel matrix $\bH_t(f_k)$ can be expressed as:
\begin{equation}
\label{eq: H_t}
\bH_t(f_k) = c_0 \ba_\text{RIS}(\bp_\BS,f_k) \underbrace{\ba^\Herm_\BS(\bp_\RIS,f_k)}_{\overset{(a)}{=}\ba^\Herm_\BS(\bp_\RIS,f_c)},
\end{equation}
where equation $(a)$ is justified by the discussion in Section~\ref{Sec: Impact of Number of Antenna in Wideband Scenario}. In addition, $\ba_\RIS(\cdot)$ and $\ba^\Herm_\BS(\cdot)$ denote the steering vectors at the \gls{RIS} and \gls{BS} in the $k$th subcarrier frequency, respectively, and both satisfy $\|\ba_\RIS(\cdot)\| = \|\ba^\Herm_\BS(\cdot)\| = 1$. Given the large \gls{LC}-\gls{RIS}, we adopt the \gls{NF} model for $\ba_\RIS(\cdot)$, as described in \eqref{eq: channel models}. Here, $\bp_\BS$ and $\bp_\RIS$ represent the center positions of the \gls{BS} and \gls{RIS}, respectively, while $c_0$ denotes the path loss coefficient of the \gls{LOS} link. With $\bomega$ fixed, the original problem P1 reduces to the following sub-problem:
\begin{subequations}
\label{eq:optimization 2}
\begin{align}
    \text {P2:}\quad&~\underset{\alpha,\bq}{\max}~\alpha
    \\&~\text {s.t.} ~~\RS(f_k)\geq \alpha,\, \forall \bp_u\in\Pset_u,\,\forall\bp_e\in\Pset_e,\,\forall k,
    \\&\quad\hphantom {\text {s.t.} } \|\bq\|_2^2\leq P_t.
\end{align}
\end{subequations}

The optimal beamformer for this sub-problem is characterized in the following lemma.
\begin{lem}
\label{lemma beamforming}
    Assuming a dominant \gls{LOS} link, blocked direct channels for both the legitimate user and the eavesdropper, and assuming that beam squinting for the \gls{BS} is negligible (green regime in Fig.~\ref{fig: N_f_fixed}), the optimal beamformer for P2 is given by $\bq=\sqrt{P_t}\ba_\BS(\bp_\RIS,f_c)$ regardless of the \glspl{MU}' and \gls{ME}'s channels.
\end{lem}
\begin{proof}
By inserting $\bH_t$ from \eqref{eq: H_t} into the secure rate expression \eqref{eq: SNR} and ignoring the direct paths $\bh_{d,g},\,g=\{u_1,\cdots,u_U,e\}$ due to blockage, $\RS$ becomes
\begin{equation}
    \!\Big[\!\log\big(\frac{1+\overbrace{|\bh_{r,u}\bGamma c_0\ba_\RIS(\bp_\BS,f_k)|^2/\sigma^2_n}^{\zeta_{u,k}(\bp_u)}\times   \overbrace{|\ba^\Herm_\BS(\bp_\RIS,f_c)\bq|^2}^{\mu}}{1+\underbrace{|\bh_{r,e}\bGamma c_0\ba_\RIS(\bp_\BS,f_k)|^2/\sigma^2_n}_{\zeta_{e,k}(\bp_e)}\times   \underbrace{|\ba^\Herm_\BS(\bp_\RIS,f_c)\bq|^2}_{\mu}}\big)\!\Big]\!^+\!\!\! .
\end{equation}
Here, the terms $\zeta_{u,k}(\bp_u)$ and $\zeta_{e,k}(\bp_e)$ remain constant for a given set of \gls{RIS} phase configurations, while the real scalar $\mu\in\Rset$ can be tuned via $\bq$. Provided that the condition $\zeta_{u,k}(\bp_u)>\zeta_{e,k}(\bp_e)$ is satisfied for all $\bp_u\in\Pset_u, \bp_e\in\Pset_e$, and $k\in\{1,\cdots,K\}$, the secure rate $\RS$ increases monotonically \gls{wrt} $\mu$. We can maximize $\mu$ to $P_t$ by defining $\bq=\sqrt{P_t}\ba_{\BS}(\bp_{\RIS},f_c)$ \cite{tse2005fundamentals}. Conversely, if $\zeta_{u,k}(\bp_u)\leq\zeta_{e,k}(\bp_e)$ for even a single combination of $(\bp_u,\bp_e,k)$ within their respective sets, the secure rate drops to $\RS=0$ regardless of the value of $\mu$. Consequently, we can set $\bq=\sqrt{P_t}\ba_{\BS}(\bp_{\RIS},f_c)$ without loss of generality, thereby completing the proof.
%where $\zeta_{u,k}(\bp_u)$ and $\zeta_{e,k}(\bp_e)$ are fixed due to the given \gls{RIS} phase shifts and $\mu\in\Rset$ is a controllable scalar in terms of $\bq$. When $\zeta_{u,k}(\bp_u)>\zeta_{e,k}(\bp_e),\,\forall \bp_u\in\Pset_u, \forall \bp_e\in\Pset_e$, $\forall k\in\{1,\cdots,K\}$, $\RS$ is monotonically increasing in $\mu$. The maximum $\mu$ is $P_t$ when we set $\bq=\sqrt{P_t}\ba_{\BS}(\bp_{\RIS},f_c)$ \cite{tse2005fundamentals}. Otherwise, when there exists at least one pair $(\bp_u,\bp_e,k)$, $\bp_u\in\Pset_u$, $\bp_e\in\Pset_e$, and $k\in\{1,\cdots,K\}$ such that $\zeta_{u,k}(\bp_u)\leq\zeta_{e,k}(\bp_e)$, then we have $\RS=0$ independent of $\mu$. Therefore, without loss of generality, we can choose $\bq=\sqrt{P_t}\ba_{\BS}(\bp_{\RIS},f_c)$ which concludes the proof.
\end{proof}
Based on the Lemma~\ref{lemma beamforming}, we fix $\bq=\sqrt{P_t}\ba_{\BS}(\bp_{\RIS},f_c)$ and in the next section, we focus on the \gls{RIS} phase shift design.

\subsection{RIS Phase Shifter Design}
\label{sec: RIS Phase Shifter Design}
Assuming a fixed beamformer, we aim to maximize the secrecy rate by optimizing the \gls{RIS} phase shifts. For large \glspl{RIS}, see Fig.~\ref{fig: N_f_fixed}, it is essential to carefully design the \gls{LC}-\gls{RIS} phase shifts across all subcarriers simultaneously. The resulting subproblem from Problem P2, while fixing \gls{BS} beamformer to configure the \gls{RIS} phase shifts, can be formulated as:
\begin{subequations}
\label{eq:optimization 3}
\begin{align}
    \text {P3:}\quad&~\underset{\alpha,\bomega_c}{\max}~\alpha
    \\&~\text {s.t.} ~~\RS(f_k)\geq \alpha,\, \forall \bp_u\in\Pset_u,\,\forall\bp_e\in\Pset_e,\,\forall k,
    \\&\quad\hphantom {\text {s.t.} } 0\leq [\bomega_c]_n < 2\pi,
    \\&\quad\hphantom {\text {s.t.} }  [\bomega(f_k)]_n = [\bomega_c]_n\beta_k, \forall n,\,\forall k.\label{eq:optimization 3 d}
\end{align}
\end{subequations}
Problem P3 is inherently non-convex due to the non-convexity of constraint (\ref{eq:optimization 3}b) with respect to $[\bomega(f_k)]_n$ for all $n$ and $k$. Without loss of generality, we omit the $[\cdot]^+$ operation from $\RS$ in \eqref{eq: SNR}, and proceed to maximize the secrecy rate using P3 by considering $\widetilde{\RS}$ instead of $\RS$, see \eqref{eq: SNR}. This simplification does not affect the solution: if the resulting $\RS > 0$, the omission is fine; otherwise, if $\widetilde{\RS} < 0$, the secrecy rate becomes zero regardless, i.e., $\RS=0$. To reformulate the constraint, we introduce a new variable $\gamma$ such that $\log(\gamma) = \alpha$. Using this substitution, constraint (\ref{eq:optimization 3}b) can be rewritten as:
\begin{equation}
\label{eq: gamma definition}
    \frac{1+\SNR_u(f_k)}{1+\SNR_e(f_k)}\geq\gamma\Rightarrow (\SNR_u(f_k)-\gamma\SNR_e(f_k))\geq\gamma-1.
\end{equation}
 Let us decompose each $\SNR_u(f_k)$ and $\SNR_e(f_k)$ in terms of a vector including exponential of \gls{RIS} phase shifts
 \begin{equation}
 \label{eq: phase shift vector}
     \bs_k\defeq[\e^{\jj[\bomega(f_k)]_1}, \cdots, \e^{\jj[\bomega(f_k)]_N}]^\Trans,\forall k.
 \end{equation}
 With this assumption, we have
\begin{align}
    \label{eq: SNR in term of s}
        &\SNR_u(f_k)=\bs_k^\Herm\bA^u_k\bs_k,\quad
        \SNR_e(f_k)=\bs_k^\Herm\bA^e_k\bs_k,\\
        \label{eq: Matrix A}
        &\bA^g_k=\frac{\diag(\bh_{r,g}(f_k)^\Herm)\bH_t(f_k)\bq\bq^\Herm\bH_t(f_k)^\Herm\diag(\bh_{r,g}(f_k))}{\sigma_n^2},
\end{align}
where $g=\{u_1,\cdots,u_U,e\}$. To solve the formulated Problem P3, we introduce two distinct approaches, each presenting unique trade-offs. The first is an \gls{SDP}-based method, which delivers excellent accuracy at the expense of high computational complexity. The second is a low-complexity method; although its performance is marginally lower than that of the \gls{SDP}-based approach, it is highly scalable and particularly well-suited for extremely large \glspl{RIS}.

\subsubsection{SDP-based method}
\label{sec: Optimization-based method}
To tackle the non-convexity of the problem in the first method, we transform P3 into a \gls{SDP} problem. Let us define $\bS_k\defeq\bs_k\bs_k^\Herm,\,\forall k$, and $\bA_{k,u,e}(\gamma)\defeq\bA_k^u(\bp_u)-\gamma\bA_k^e(\bp_e)$ where $\bA_{k,u,e}(\gamma)$ is a function of $\bp_u,\,\bp_e,\,f_k$ and $\gamma$, but we drop $\bp_u$, $\bp_e$ for notational simplicity. By considering $\bS_c$ as the matrix associated to the center frequency $f_c$, (\ref{eq:optimization 3 d}) changes to $\bS_k=\bS_c^{\circ\beta_k}, \forall k$, where the notation $(\cdot)^{\circ\beta_k}$ denotes the element-wise power (Hadamard power). After applying these reformulations in P3, and because the logarithm function is increasing monotonically, the problem P3 can be changed to P4 in the following:
 \begin{subequations}
\label{eq:optimization 4}
\begin{align}
    \text{P4:}&~\underset{\gamma,\bS_c}{\max}~\gamma
    \\&~\text {s.t.} ~~\text{C1: }\tr(\bA_{k,u,e}(\gamma)\bS_c^{\circ\beta_k})\!\geq\!\gamma-1, \forall (\bp_u,\bp_e)\!\in\!\Pset_u\!\times\!\Pset_e,\,\forall k,
    \\&\quad\hphantom {\text {s.t.} }\text{C2: } \bS_k\succeq 0,\,\forall k,
    \\&\quad\hphantom {\text {s.t.} }\text{C3: } \rank(\bS_k)=1,\,\forall k, \\&\quad\hphantom {\text {s.t.} }\text{C4: }\diag(\bS_k)=\bone_N,\,\forall k.
\end{align}
\end{subequations}
This problem is still non-convex due to the non-convexity in
C1 and C3 in $\bS_k,\forall k$, and being coupled $\gamma$ and $\bS_k,\forall k$ in C1. However, before the resolving these non-convexities, first, we introduce Lemma~\ref{lemma S for all k} to simplify P4 and then solve it.

\begin{lem}
\label{lemma S for all k}
   If constraints C2, C3, and C4 hold for the center frequency, then they hold for all $f_k,\, \forall k$.
\end{lem}

\begin{proof}
Recall $\bS_k=\bS_c^{\circ\beta_k}, \forall k$, we begin with the rank-one constraint in C3. If $\bS_c$ is rank one, it can be written as $\bS_c = \bs_c \bs_c^\Herm$. Then, we have:
\[
\bS_k = \bS_c^{\circ \beta_k} = (\bs_c \bs_c^\Herm)^{\circ \beta_k} \overset{(a)}{=} \bs_c^{\circ \beta_k} (\bs_c^{\circ \beta_k})^\Herm,
\]
where $(a)$ is valid for rank one matrix and the result shows $\bS_k,\,\forall k$, is also rank one. Next, we verify that $\bS_k$ is positive semidefinite (C2), assuming $\bS_c \succeq 0$. Since $\bS_c \succeq 0$ and $\rank(\bS_c)=1$, we can rewrite $\bS_c$ as $\bS_c = e_c\bx_c\bx_c^\Herm,$ where $e_c>0$ is the only non-zero eigenvalue of $\bS_c$ and $\bx_c$ is corresponding eigenvector. Then:
\[
\bS_k = \bS_c^{\circ \beta_k} = (e_c\bx_c\bx_c^\Herm)^{\circ \beta_k} = e_c^{\beta_k}\bx_c^{\circ \beta_k}(\bx_c^\Herm)^{\circ \beta_k},
\]
where $e_c^{\beta_k}>0$ and confirms that $\bS_k$ is also positive semidefinite. Finally, consider the diagonal constraint in C4. If $\diag(\bS_c) = \bI_N$, then $[\bS_c]_{i,i} = 1,\,\forall i$, and thus $[\bS_c]_{i,i}^{\beta_k} = 1^{\beta_k}=1,\,\forall k$. Hence, $\diag(\bS_k) = \bI_N,\, \forall k$. This completes the proof.
\end{proof}

Using Lemma~\ref{lemma S for all k}, we can rewrite problem P4 as follows:
 \begin{equation}
\label{eq:optimization 5}
\begin{aligned}
    \text {P5:}&~\underset{\gamma,\bS_c}{\max}~\gamma
    \\&~\text {s.t.} ~~\text{C1},\, \widehat{\text{C2}}: \bS_c\succeq 0,\,
    \\&\quad\hphantom {\text {s.t.} } \widehat{\text{C3}}: \rank(\bS_c)=1,\,  \widehat{\text{C4}}: \diag(\bS_c)=\bone_N.
\end{aligned}
\end{equation}
Note that the impact of the frequency dependent component is now isolated in C1. This problem is still non-convex due to the non-convexity in C1 and $\widehat{\text{C3}}$ in $\bS_c$, and being coupled $\gamma$ and $\bS_c$ in C1. We resolve the non-convexity of each one in the following.
\paragraph{Rank one constraint in \texorpdfstring{$\widehat{\text{C3}}$}{C3 hat}} \mbox{}
 To tackle this issue, we adopt the penalty method exploited in \cite{delbari2024far}. The basic idea is to replace the rank constraint with inequality $\|\bS_c\|_*-\|\bS_c\|_2 \leq 0$, which holds only if $\bS_c$ has rank smaller or equal to one. While the new constraint is still non-convex, one can apply the first-order Taylor approximation to make it convex. Let $\bS_c^{(i)}$ denotes the value of matrix $\bS_c$ in the $i$th iteration. According to the first-order Taylor approximation:
\begin{equation}
\label{eq: taylor approximation}
    \|\bS_c\|_2\geq\|\bS_c^{(i)}\|_2+\tr\big(\blambda_{\max}(\bS_c^{(i)})\blambda_{\max}^\Herm(\bS_c^{(i)})(\bS_c-\bS_c^{(i)})\big).
\end{equation}
By adopting the penalty method and applying \eqref{eq: taylor approximation} into the cost function of P5, we have
\begin{subequations}
\label{eq:optimization 6}
\begin{align}
    \text {P6:}\quad&~\underset{\gamma,\bS_c}{\max}~\gamma-\eta^{(i)}\Big(\|\bS_c\|_*-\|\bS_c^{(i)}\|_2-\tr\big(\blambda_{\max}(\bS_c^{(i)}),\nonumber
    \\&\quad\quad\quad\times\blambda_{\max}^\Herm(\bS_c^{(i)})(\bS_c-\bS_c^{(i)})\big)\Big)
    \\&~\text {s.t.} ~~\text{C1}, \widehat{\text{C2}}, \widehat{\text{C4}}.
\end{align}
\end{subequations}
Here, $\eta^{(i)}$ is the penalty factor at iteration $i$, which increases gradually. By selecting a sufficiently large $\eta$, problems P5 and P6 become equivalent. In the following, we will address the non-convexity of C1 in $\bS_c$.

\paragraph{Hadamard power constraint on \texorpdfstring{$\bS_c$}{Sc} in C1} \mbox{}
 To tackle with this issue, we substitute $\bS_c^{\circ\beta_k}$ with its first order Taylor approximation\footnote{Higher-order Taylor approximations are also applicable; however, they increase computational complexity. We choose the first-order approximation to maintain simplicity.} as follows:
\begin{equation}
    \bS_c^{\circ\beta_k}\approx{\left(\bS_c^{(i)}\right)}^{\circ(\beta_k)}+\left(\beta_k{\left(\bS_c^{(i)}\right)}^{\circ(\beta_k-1)}\right)\odot(\bS_c-\bS_c^{(i)}),\,\forall k,
\end{equation}
where $i$ denotes $i$th iteration. By substituting this approximation in C1, we can define $\widehat{\text{C1}}$ as follows:
\begin{align}
&\tr\!\left(\!\!\bA_{k,u,e}(\gamma)\left(\!\!{\left(\bS_c^{(i)}\right)}^{\circ(\beta_k)}\!\!\!\!+\!\!\left(\beta_k{\left(\bS_c^{(i)}\right)}^{\circ(\beta_k-1)}\right)\odot\left(\bS_c-\bS_c^{(i)}\right)\!\!\right)\!\!\right)\nonumber\\
&\geq\gamma-1, \forall (\bp_u,\bp_e)\in\Pset_u\times\Pset_e,\,\forall k.
\end{align}
By substituting $\widehat{\text{C1}}$ instead of C1 in P6, we have:
\begin{subequations}
\label{eq:optimization 7}
\begin{align}
    \text {P7:}\quad&~\underset{\gamma,\bS_c}{\max}~\gamma-\eta^{(i)}\Big(\|\bS_c\|_*-\|\bS_c^{(i)}\|_2-\tr\big(\blambda_{\max}(\bS_c^{(i)}),\nonumber
    \\&\quad\quad\quad\times\blambda_{\max}^\Herm(\bS_c^{(i)})(\bS_c-\bS_c^{(i)})\big)\Big)
    \\&~\text {s.t.} ~~\widehat{\text{C1}}, \widehat{\text{C2}}, \widehat{\text{C4}}.
\end{align}
\end{subequations}
\paragraph{\texorpdfstring{Coupled $\gamma$ and $\bS_c$ in $\widehat{\text{C1}}$}{Coupled gamma and Sc in C1-hat}} \mbox{}
To address the coupling of $\gamma$ and $\bS_c$ in $\widehat{\text{C1}}$, we employ \gls{AO}, where one variable is fixed while the other is optimized. On one hand, when $\gamma$ is fixed, the optimization problem P7 becomes convex in the
matrix $\bS_c$ because the objective function is concave, and the
constraints define a convex set. Therefore, it can be efficiently
solved using standard convex optimization solvers such as
CVX \cite{cvx}. On the other hand, when $\bS_c$ is fixed, The problem P7 is linear in terms of $\gamma$ and its closed-form solution is given
by:
\begin{equation}
\label{eq: best gamma}
    \gamma=\underset{\forall \bp_u\in\Pset_u,\,\forall\bp_e\in\Pset_e,\,\forall k}{\min}~\frac{\tr(\bA_k^u(\bp_u)\bS_k)+1}{\tr(\bA_k^e(\bp_e)\bS_k)+1}.
\end{equation}

The proposed \gls{SDP}-based algorithm is summarized in Algorithm \ref{alg:cap}. It consists of two loops; the inner loop finds a rank-one solution to Problem P7, and the outer loop applies successive convex approximation (SCA).

\begin{algorithm}[t]
\caption{Proposed \gls{SDP}-based Algorithm for \gls{LC}-\gls{RIS} phase-shift design}\label{alg:cap}
\begin{algorithmic}[1]
\STATE \textbf{Initialize:} $\bs_c^{(0)}=\e^{\jj2\pi\times\mathrm{rand}(N)},\bS_c^{(0)}=\bs_c^{(0)}{\bs_c^{(0)}}^\Herm$, $\gamma^{(0)}$.
\FOR{$j=1, \cdots, J_{\tmax}$}
\STATE Calculate $\bA_{k,u,e}(\gamma^{(j-1)}),\forall (\bp_u,\bp_e)\in\Pset_u\times\Pset_e,\,\forall k$.
    %\WHILE{$\|\bS_c^{(i)}-\bS_c^{(i-1)}\|_F^2\geq\epsilon_1$ and $i\leq I_{\tmax}$}
    \FOR{$i=1, \cdots, I_{\tmax}$}
    \STATE Solve convex P7 for given $\bS_c^{(i-1)}$, and store the intermediate solution $\bS_c^{(i)}$. Then update $\eta^{(i)} =5\eta^{(i-1)}$.
    \ENDFOR
    \STATE Set $\bS_k=\left(\bS_c^{(I_{\max})}\right)^{\circ(\beta_k)},\,\forall k,$ and $\bS_c^{(0)}=\bS_c^{(I_{\max})}$.
    \STATE Calculate $\gamma^{(j)}$ from \eqref{eq: best gamma}.
    
    \ENDFOR
\end{algorithmic}
\end{algorithm}

\subsubsection{Low-complexity method}
\label{sec: Analytical method}
The \gls{SDP}-based algorithm descussed in Section \ref{sec: Optimization-based method} is computationally complex, see Section~\ref{sec: Algorithm and Complexity Analysis} for a complexity analysis. Therefore, in this section, we design a low-complexity scheme. To tackle the non-convexity of Problem~P3 established by the definitions in \eqref{eq: gamma definition}, \eqref{eq: phase shift vector}, and \eqref{eq: SNR in term of s}, we formulate the following problem:
 \begin{subequations}
\label{eq:optimization 8}
\begin{align}
    \text{P8:}&~\underset{\gamma,\bs_c}{\max}~\gamma
    \\&~\text {s.t.} ~~\text{C1: }\frac{1+\SNR_e(f_k)}{1+\SNR_u(f_k)}\geq\gamma, \forall (\bp_u,\bp_e)\!\in\!\Pset_u\!\times\!\Pset_e,\,\forall k,
    \\&\quad\hphantom {\text {s.t.} }\text{C2: } \bs_k=\bs_c^{\circ\beta_k},\,\forall k,
    \\&\quad\hphantom {\text {s.t.} }\text{C3: } |[\bs_c]_n|=1,\,\forall n.
\end{align}
\end{subequations}
We optimize $\gamma$ and $\bs_c$ iteratively. Even for a fixed $\gamma$, this problem remains non-convex due to the unit-modulus constraint C3. To render the problem tractable, we introduce and maximize a surrogate lower bound \cite{Shen2019}. Let us define 
\begin{align}
\label{eq: calculation matrix A}
    \bA_{k,u,e}(\gamma)&\defeq\bA^u_k(\bp_u)-\gamma\bA^e_k(\bp_e),\\
    \bPhi_{k,u,e}(\gamma)&\defeq\bA_{k,u,e}(\gamma)-\lambda_{\min}(\bA_{k,u,e}(\gamma))\bI_N,
    \label{eq: calculation matrix Phi}
\end{align}
where $\lambda_{\min}(\cdot)$ denotes the minimum eigenvalue of a matrix. This eigenvalue is inherently negative, stemming from the $-\gamma\bA^e_k(\bp_e)$ term in \eqref{eq: calculation matrix A}. Based on these definitions, we can reformulate constraint C1 given by $\bs_k^\Herm\bA_{k,u,e}(\gamma)\bs_k-\gamma+1\geq0$ into a lower bound using the following lemma.
\begin{lem}
\label{lem: lower bound}
    A valid lower bound for $\bs_k^\Herm\bA_{k,u,e}(\gamma)\bs_k$ is
    \begin{equation}
        \lambda_{\min}(\bA_{k,u,e}(\gamma)) N+2\real\{\bs_k^\Herm\bbeta_{k,u,e}(\gamma)\}-\tilde{\bs}_k^\Herm\bPhi_{k,u,e}(\gamma)\tilde{\bs}_k,
    \end{equation}
    where
    \begin{equation}
        \label{eq: beta calculation}
        \bbeta_{k,u,e}(\gamma)\defeq\bPhi_{k,u,e}(\gamma)\tilde{\bs}_k.
    \end{equation}
 Here, $\gamma$ and the constant vector $\tilde{\bs}_k$ are fixed, and equality is achievable by $\bs_k=\tilde{\bs}_k,\,\forall k$.
\end{lem}
\begin{proof}
Because $\lambda_{\min}(\bA_{k,u,e}(\gamma))$ represents the minimum eigenvalue, subtracting it multiplied by the identity matrix in \eqref{eq: calculation matrix Phi} ensures that all eigenvalues of $\bPhi_{k,u,e}(\gamma)$ are non-negative. Consequently, $\bPhi_{k,u,e}(\gamma)$ is a positive semi-definite matrix, which implies that $(\bs_k-\tilde{\bs}_k)^\Herm\bPhi_{k,u,e}(\gamma)(\bs_k-\tilde{\bs}_k)\geq0$. By expanding this inequality and substituting the definition of $\bPhi_{k,u,e}(\gamma)$, we obtain $\bs_k^\Herm\bA_{k,u,e}(\gamma)\bs_k-\underbrace{\bs_k^\Herm\lambda_{\min}(\bA_{k,u,e}(\gamma))\bI_N\bs_k}_{\lambda_{\min}(\bA_{k,u,e}(\gamma))\|\bs_k\|^2}-2\real\{\bs_k^\Herm\bbeta_{k,u,e}(\gamma)\}+\tilde{\bs}_k^\Herm\bPhi_{k,u,e}(\gamma)\tilde{\bs}_k\geq0$. Rearranging these terms yields the proposed lower bound, concluding the proof.
\end{proof}
With the help of Lemma~\ref{lem: lower bound}, we can reformulate Problem~P8 by maximizing its surrogate function as follows:
\begin{subequations}
\label{eq:optimization 9}
\begin{align}
    \text{P9:}&~\underset{\bs_c}{\max}~\underset{k,\bp_u,\bp_e}{\min}~\real\{\bs_k^\Herm\bbeta_{k,u,e}(\gamma)\}
    \\&\quad\hphantom {\text {s.t.} }\text{C1: } \bs_k=\bs_c^{\circ\beta_k},\,\forall k,
    \\&\quad\hphantom {\text {s.t.} }\text{C2: } |[\bs_c]_n|=1,\,\forall n.
\end{align}
\end{subequations}
Problem~P9 can be solved iteratively. A naive approach would be to fix the solution in previous step $\tilde{\bs}_k$ at each iteration and optimize $\bs_k$ exclusively for the specific subcarrier and location that yields the minimum secrecy rate. However, this hard-minimum approach is highly susceptible to the ping-pong effect, where the algorithm continuously oscillates between a finite set of active constraints and fails to converge. To address this issue, rather than directly maximizing the non-smooth worst-case \gls{RS} in Problem~P9, we maximize its \gls{LSE} approximation. This approach provides a smooth and differentiable surrogate objective.
\begin{lem}
    \label{lem: LSE}
    A valid lower bound for $\underset{n}{\min}~x_n$ is given by
    \begin{equation}
        \label{eq: lem LSE}
        \underset{n}{\min}~x_n \geq -\mu\log\left(\sum_n \exp\left(-\frac{x_n}{\mu}\right)\right),
    \end{equation}
    where $\mu > 0$ is a smoothing parameter. This approximation becomes exact (i.e., equality is achieved) as $\mu\to0$.
\end{lem}
\begin{proof}
Since the exponential function is strictly positive, it follows that $\sum_n \exp(-\frac{x_n}{\mu})\geq \exp(-\frac{\underset{n}{\min}~x_n}{\mu}),\,\forall \mu>0$. Because the natural logarithm is a monotonically increasing function, applying $\log(\cdot)$ to both sides yields $\log\left(\sum_n \exp\left(-\frac{x_n}{\mu}\right)\right) \geq -\frac{\underset{n}{\min}~x_n}{\mu}$. Multiplying both sides by $-\mu$ reverses the direction of the inequality, which directly results in \eqref{eq: lem LSE}.
\end{proof}
Applying Lemma~\ref{lem: LSE}, we obtain:
\begin{subequations}
    \begin{align}
    \label{eq: LSE lower bound}
    &\underset{k,\bp_u,\bp_e}{\min}\widetilde{\RS}(f_k,\bp_u,\bp_e) \geq -\mu\log(w_\text{total}),\\
&w_\text{total}\defeq\sum_{k=1}^{K}\sum_{\bp_u\in\Pset_U}\sum_{\bp_e\in\Pset_e}w(k,\bp_u,\bp_e),\\
        &w(k,\bp_u,\bp_e)\defeq\exp(-\frac{\widetilde{\RS}(f_k,\bp_u,\bp_e)}{\mu}),
    \end{align}
\end{subequations}
where $\widetilde{\RS}(f_k,\bp_u,\bp_e)$ is evaluated using the phase-shift vector $\tilde{\bs}_k$ from the previous iteration. By assigning exponentially heavier weights to the scenarios with the lowest \gls{RS}, the optimal phase-shift vector is derived as:
\begin{equation}
\label{eq:LSE}
\bs_c=\exp\!\!\left(\jj \arg\!\left(\sum_{k=1}^{K}\sum_{\bp_u\in\Pset_U}\sum_{\bp_e\in\Pset_e}\!\!\!\frac{w(k,\bp_u,\bp_e)}{w_\text{total}}\bbeta_{k,u,e}(\gamma)\!\right)\!\!\right).
\end{equation}
We execute this procedure iteratively. At each iteration, the previous solution vectors $\tilde{\bs}_k$ are updated with the newly computed $\bs_k$ for all $k$. Once the sequence of phase-shift vectors $\bs_k, \forall k$, converges, we update the auxiliary variable $\gamma$ using \eqref{eq: best gamma}. The proposed algorithm is summarized in Algorithm~\ref{alg:cap 2}.

\subsection{Algorithm and Complexity Analysis}
\label{sec: Algorithm and Complexity Analysis}
For every iteration of Algorithm~\ref{alg:cap}, evaluating the nuclear norm in line 5 is the most computationally demanding operation, scaling with a complexity of $\bigO(N^3)$. The number of different constraints generated by $\widehat{\text{C1}}$ is the bottleneck and proportional to $|\Pset_u||\Pset_e|K$. Thus, the complexity of the Algorithm \ref{alg:cap} in total is $\bigO(I_{\max}|\Pset_u||\Pset_e|N^3K)$.

In contrast, Algorithm~\ref{alg:cap 2} exhibits significantly lower computational complexity. The most challenging operation is computing the minimum eigenvalue in \eqref{eq: calculation matrix Phi}. Because our phase-shift design operates in the \gls{mmWave} regime, where \gls{LOS} paths dominate, we focus exclusively on the \gls{LOS} links for both the user and the eavesdropper. This implies that the matrix $\bA_{k,u,e}$ has a rank of at most two effectively for any specific user and eavesdropper location pair. Let us expand $\bA^u_k=\ba_k^u(\ba_k^u)^\Herm$ and $\bA_k^e=\ba^e_k(\ba_k^e)^\Herm$. We can then decompose $\bA_{k,u,e}$ as follows:
\begin{equation}
    \bA_{k,u,e}=\begin{bmatrix}
\ba_k^u & \ba_k^e
\end{bmatrix}
\begin{bmatrix}
    1 & 0 \\
    0 & -\gamma
\end{bmatrix}\begin{bmatrix}
    (\ba_k^u)^\Herm \\
    (\ba_k^e)^\Herm
\end{bmatrix}.
\end{equation}
To reduce complexity, we define an auxiliary matrix $\bC$:
\begin{equation}
    \bC=\begin{bmatrix}
    1 & 0 \\
    0 & -\gamma
\end{bmatrix}\begin{bmatrix}
    (\ba_k^u)^\Herm\ba_k^u & (\ba_k^u)^\Herm\ba_k^e \\
    (\ba_k^e)^\Herm\ba_k^u & (\ba_k^e)^\Herm\ba_k^e
\end{bmatrix}.
\end{equation}
The second matrix in this product is independent of both $\gamma$ and $\bs_c$, allowing it to be precomputed outside the iterative loops. Since $\bC$ is now a $2\times2$ matrix, its eigenvalues can be efficiently calculated by solving the characteristic polynomial $\lambda^2-\tr\{\bC\}\lambda+\det(\bC)=0$. Thus, the complexity of the Algorithm \ref{alg:cap} in total is linear in $N$ $\bigO(T_{\max}I_{\max}|\Pset_u||\Pset_e|NK)$.

\begin{algorithm}[t]
\caption{Proposed Algorithm for Problem P8}\label{alg:cap 2}
\begin{algorithmic}[1]
\STATE Calculate $\bA_k^u(\bp_u)$ and $\bA_k^e(\bp_e),\, \forall \bp_u\in\Pset_u,\,\forall\bp_e\in\Pset_e,\,\forall k$ via \eqref{eq: Matrix A}. Initialize $\RS_\text{final}=-1000$.
\FOR{$t=1, \cdots, T_{\tmax}$}
\STATE \textbf{Initialize:} $\tilde{\bs}_c=\e^{\jj2\pi\times\mathrm{rand}(N)}, \tilde{\bs}_k=\tilde{\bs}_c^{\circ\beta_k},\,\forall k$, $\gamma^{(0)}$.
\FOR{$j=1, \cdots, J_{\tmax}$}
\FOR{$i=1, \cdots, I_{\tmax}$}
\STATE Calculate $\bA_{k,u,e}(\gamma^{(j-1)}),\bPhi_{k,u,e}(\gamma^{(j-1)}),$ and $\bbeta_{k,u,e}(\gamma^{(j-1)}),\,\forall (\bp_u,\bp_e)\in\Pset_u\times\Pset_e,\,\forall k$ via \eqref{eq: calculation matrix A}, \eqref{eq: calculation matrix Phi}, and \eqref{eq: beta calculation}, respectively.
    \STATE Calculate $\bs_c$ via \eqref{eq:LSE} and set $\tilde{\bs}_k=\bs_c^{\circ\beta_k},\,\forall k$.
    \IF{$\RS_\text{final}<\underset{k,\bp_u,\bp_e}{\min}~\RS$}
    \STATE $\RS_\text{final}=\underset{k,\bp_u,\bp_e}{\min}~\RS$,
    $\bs_\text{final}=\bs_c$.
    \ENDIF
    \ENDFOR
    \STATE Set $\bs_k=\bs_\text{final}^{\circ\beta_k},\,\forall k$ and calculate $\gamma^{(j)}$ from \eqref{eq: best gamma} when $\bS_k=\bs_k\bs_k^\Herm,\,\forall k$.
    
    \ENDFOR
    \ENDFOR
\end{algorithmic}
\end{algorithm}

\section{Performance Evaluation}
\label{sec: Performance Comparison}
\subsection{Simulation Setup}
We adopt the simulation configuration for coverage extension illustrated in Fig. \ref{fig:system model}, where the \gls{RIS} center is the origin of the Cartesian coordinate system, i.e., $[0,0,0]~\text{m}$. A set of legitimate users is uniformly distributed within a fixed region defined as $\Pset_u\in\{(\x,\y,\z):5~\text{m}\leq\x\leq 7~\text{m}, 0~\text{m}\leq\y\leq 2~\text{m}, \z=-5\}$. We also estimated a fixed area for eavesdropper in $\Pset_e\in\{(\x,\y,\z):5~\text{m}\leq\x\leq 6~\text{m}, -2~\text{m}\leq\y\leq -1~\text{m}, \z=-5\}$. The \gls{BS} is equipped with an $16 \times 16 = 256$-element \gls{UPA}, arranged along the $x$–$z$ plane and located at $[10, 10, 5]~\text{m}$. The \gls{RIS} is modeled as a \gls{ULA} with $N_y = 100$ elements aligned along the $y$-axis. The inter-element spacing for both the \gls{BS} and \gls{RIS} arrays is set to half the wavelength, $\lambda_c=5~$mm. The noise power is calculated as $\sigma_n^2=W_kN_0N_{\rm f}$  with $N_0=-174$~dBm/Hz, $N_{\rm f}=6$~dB, and bandwidth of each subcarrier $W_k=4.2~$MHz. We choose $60$~GHz carrier frequency and the total bandwidth $W=8.64~$GHz in accordance with the IEEE 802.11ay standard \cite{ieee80211ay}. The pathloss for the $k$th subcarrier is modeled as $\rho_k(d_0/d)^\sigma$, where $\rho_k=(\frac{c}{4\pi f_k})^2$ at $d_0=1$~m, and $c$ is the speed of light in vacuum. Moreover, we adopt the pathloss exponent $\sigma = (2,2,2)$ and Ricean $K$-factors in \eqref{eq: channel model} $\bar{k}_r=\tilde{k}_r=(0,0.1,0.1),\,\forall r$, for the \gls{BS}-\gls{MU}, \gls{BS}-\gls{RIS}, and \gls{RIS}-\gls{MU} channels, respectively, and $R=10$. 
To evaluate performance, we compare the proposed method against three benchmarks. These approaches do not account for the frequency-dependent phase response of individual \gls{LC}-\gls{RIS} elements:
\begin{itemize}
    \item \textbf{Benchmark 1 (All subcarriers, Area):} The \gls{RIS} phase shifts are optimized to distribute reflected power across the entire coverage area for all subcarrier frequencies while suppressing power in the eavesdropper region.
    \item \textbf{Benchmark 2 (Center frequency, Area):} The \gls{RIS} phase shifts are designed to maximize coverage area power and minimize eavesdropper region power, but only at the center frequency $f_c = 60$~GHz (Similar algorithms in \cite{Xing2024,Gu2023,delbari2024far}).
    \item \textbf{Benchmark 3 (All subcarriers, Point):} The \gls{RIS} phase shifts are optimized across all subcarrier frequencies, but only based on the exact locations of users and the eavesdropper, without explicitly shaping the beam across a broader area (Similar algorithms in \cite{Jiang2021,Arshad2025}).
\end{itemize}

 The other parameters used in the simulations are as follows: $P_t=10~$dBm, $\beta=2.4$ \cite{neuder2024architecture}, $\eta^{(0)}=0.0018$, and $T_{\max}=10$ are assumed. In addition, $J_{\max}=2$, $I_{\max}=9$ and $J_{\max}=14$, $I_{\max}=50$ are adopted for the \gls{SDP}-based and scalable algorithms, respectively.

 \begin{remk}
The MATLAB codes used to generate the simulation results in this section are publicly available online at:\\ \href{https://github.com/MohamadrezaDelbari/LC-RIS-Wideband}{\textcolor{blue}{https://github.com/MohamadrezaDelbari/LC-RIS-Wideband}}
 \end{remk}

\subsection{Simulation Results}
\begin{figure}
    \centering
    \includegraphics[width=0.4\textwidth]{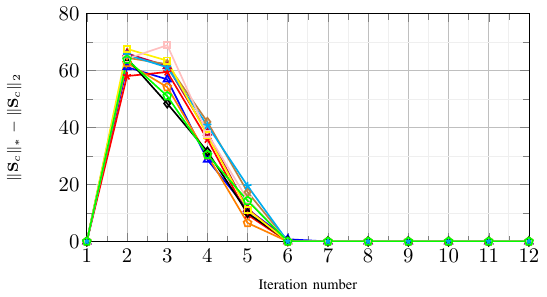}
    \caption{Convergence behavior of Algorithm \ref{alg:cap} presenting the rank constraint in P6 with different initializations versus the iteration number.}
    \label{fig:rank}
    \vspace{-5 mm}
\end{figure}

\subsubsection{Convergence and complexity comparison}

Fig.~\ref{fig:rank} illustrates the convergence behavior of Algorithm~\ref{alg:cap}. Because the algorithm is initialized with a rank-one matrix, the initial value of $\|\bS_c\|_*-\|\bS_c\|_2$ is zero, confirming that $\rank(\bS_c)=1$. As the iterations progress, this metric initially increases before ultimately decreasing, converging back to a rank-one solution that is feasible for Problem P7.

Fig.~\ref{fig: convergence} illustrates the convergence behavior of Algorithm~\ref{alg:cap 2} across 10 distinct random initializations. Specifically, Fig.~\ref{fig: convergence a} plots the \gls{LSE} surrogate function ($\widetilde{\RS}$) defined in \eqref{eq: LSE lower bound}, while Fig.~\ref{fig: convergence b} depicts the actual minimum secrecy rate ($\RS$) over the number of iterations.  As observed, the majority of these initializations successfully improve the secrecy rate. Although there is no strict guarantee of a monotonically non-decreasing \gls{RS} since \eqref{eq:LSE} provides a sub-optimal rather than a global solution to Problem P9, the algorithm effectively converges to a high-quality local optimum (plotted by the cyan curve) that significantly enhances the overall secrecy rate.

Fig.~\ref{fig: complexity} illustrates the execution times of the two proposed algorithms (the \gls{SDP}-based approach and the scalable method)\footnote{The algorithms were implemented in MATLAB R2024a and executed on an Arch Linux system equipped with an AMD Ryzen 9 7950X (16-core) CPU and 64 GB of RAM.}. As shown, the scalable method requires significantly less execution time compared to the \gls{SDP} method.  As anticipated, the computational complexity of the \gls{SDP} method scales cubically, approximating $\bigO(N^3)$, whereas the scalable solution exhibits a linear growth rate of $\bigO(N)$. For example, when $N=300$, the required time for the \gls{SDP} approach is about 1 hour, while this is less than a minute for the scalable approach.

\begin{figure}
    \centering
    \begin{subfigure}{0.4\textwidth}
    \includegraphics[width=\textwidth]{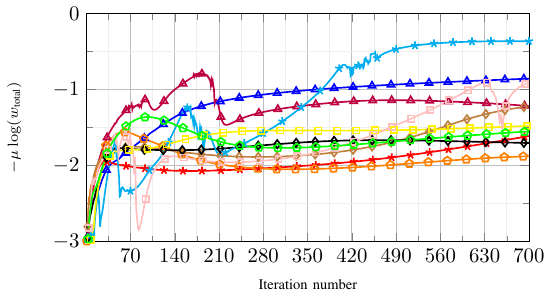}
    \caption{Convergence behavior of \gls{LSE} function.}
    \label{fig: convergence a}
    \end{subfigure}
    \begin{subfigure}{0.4\textwidth}
\includegraphics[width=\textwidth]{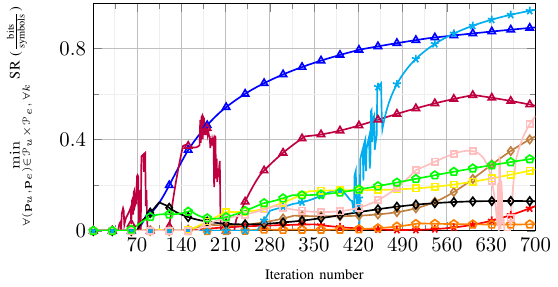}
    \caption{Minimum secrecy rate versus iteration number.}
    \label{fig: convergence b}
    \end{subfigure}
    \caption{Convergence behavior of Algorithm \ref{alg:cap 2} with different initializations versus the iteration number.}
    \label{fig: convergence}
    \vspace{- 5 mm}
\end{figure}

\begin{figure}
    \centering
    \includegraphics[width=0.4\textwidth]{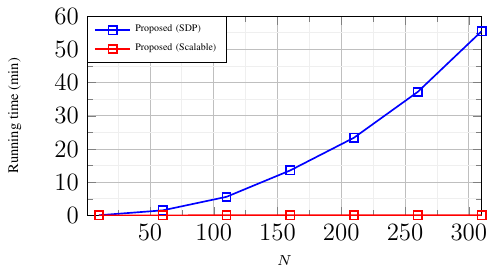}
    \caption{Running time of both proposed algorithms versus the number of the \gls{RIS} elements $(N)$.}
    \label{fig: complexity}
    \vspace{- 5 mm}
\end{figure}

\begin{figure}
    \centering
    \includegraphics[width=0.4\textwidth]{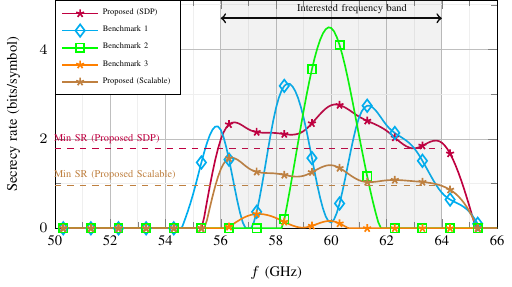}
    \caption{The minimum secrecy rate over all possible locations $\bp_u\in\Pset_u$ and $\bp_e\in\Pset_e$ versus frequency.}
    \label{fig:SR_f}
    \vspace{-5mm}
\end{figure}

\begin{figure*}
\centering
\begin{subfigure}{0.19\textwidth}
    \caption{Proposed SDP, 56 GHz}
\includegraphics[width=\textwidth,height=0.7\textwidth]{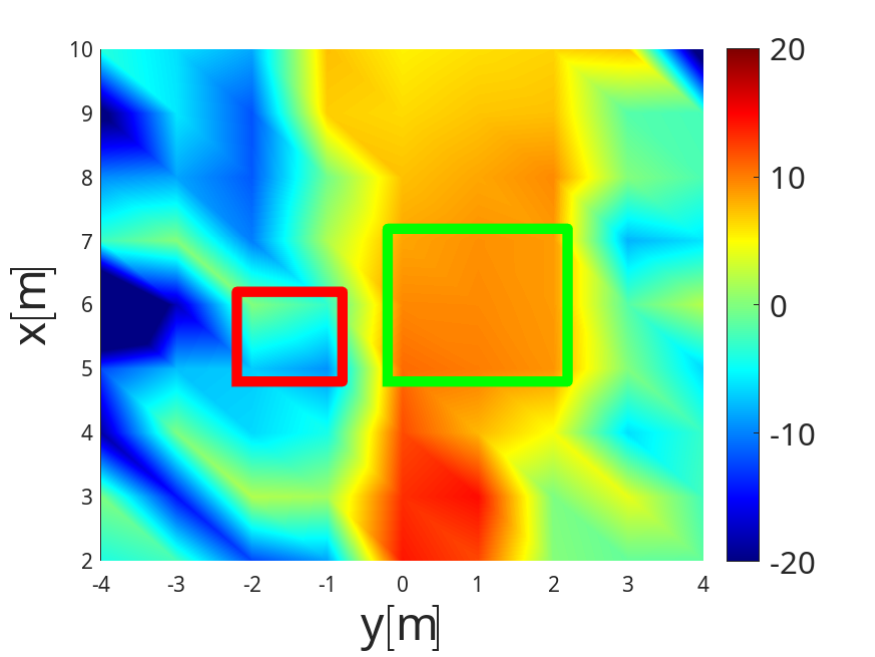}
    \label{fig: Proposed_56}
    \vspace{-5 mm}
\end{subfigure}
\begin{subfigure}{0.19\textwidth}
    \caption{Proposed scalable, 56 GHz}
    \includegraphics[width=\textwidth,height=0.7\textwidth]{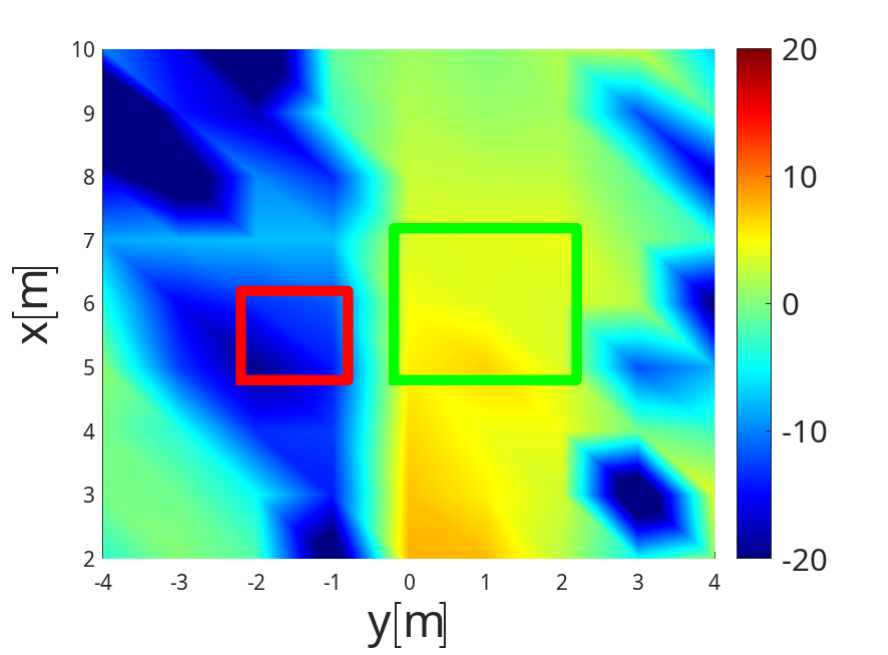}
    \label{fig: Proposed_S_56}
    \vspace{-5 mm}
\end{subfigure}
\begin{subfigure}{0.19\textwidth}
    \caption{Benchmark 1, 56 GHz}
   \includegraphics[width=\textwidth,height=0.7\textwidth]{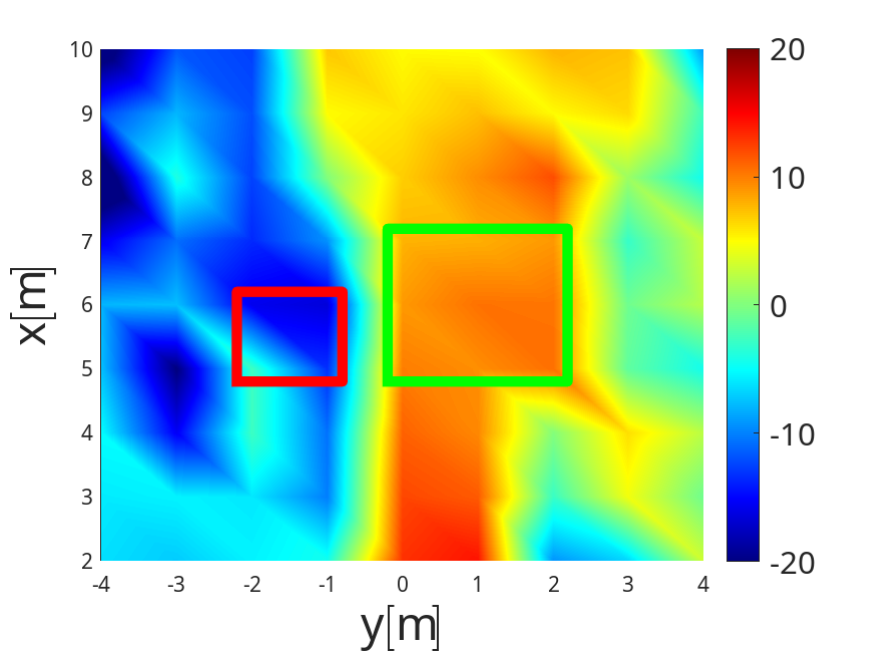}
    \label{fig: Benchmark 1_56}
    \vspace{-5 mm}
\end{subfigure}
\begin{subfigure}{0.19\textwidth}
    \caption{Benchmark 2, 56 GHz}
    \includegraphics[width=\textwidth,height=0.7\textwidth]{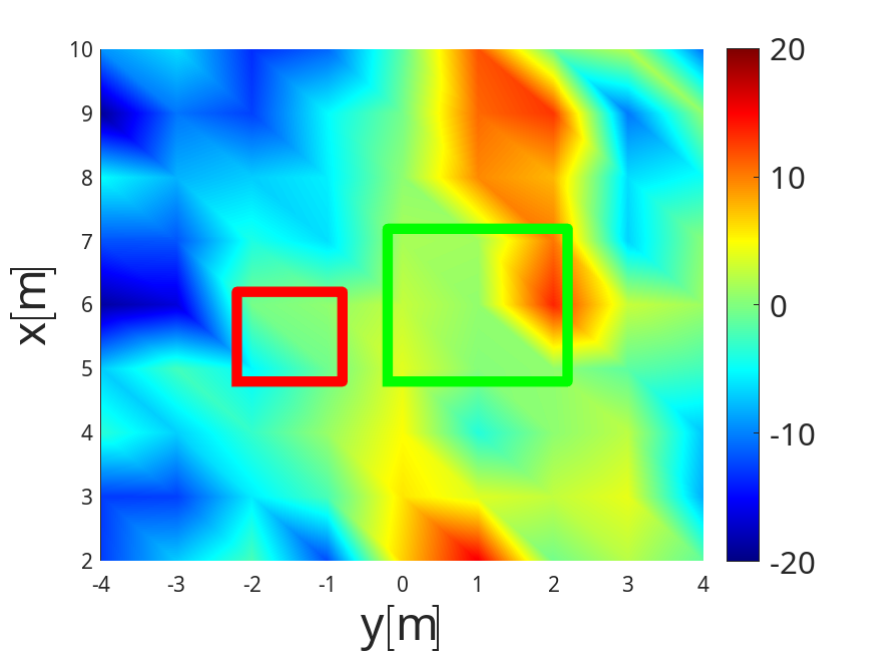}
    \label{fig: Benchmark 2_56}
    \vspace{-5 mm}
\end{subfigure}
\begin{subfigure}{0.19\textwidth}
    \caption{Benchmark 3, 56 GHz}
    \includegraphics[width=\textwidth,height=0.7\textwidth]{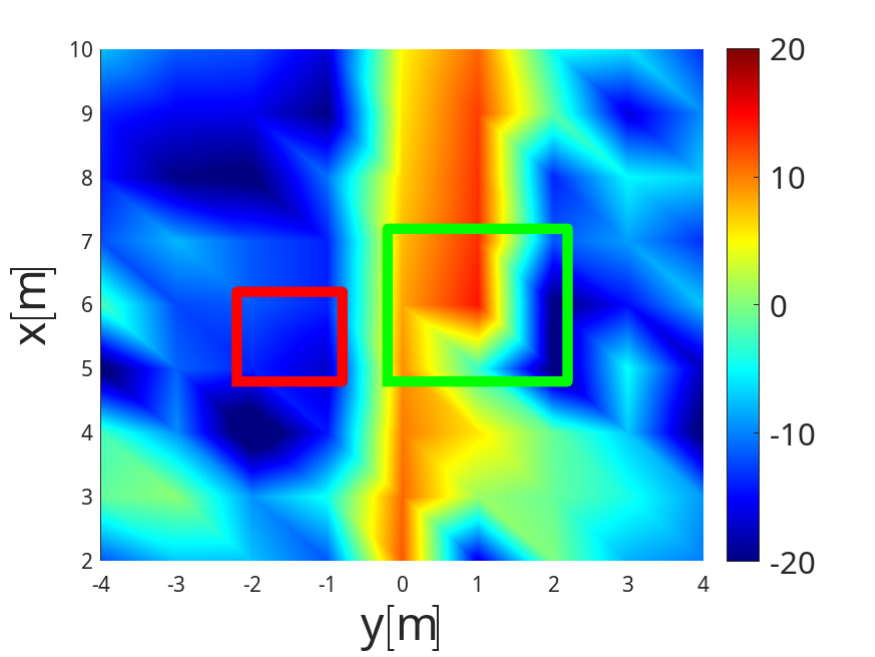}
    \label{fig: Benchmark 3_56}
    \vspace{-5 mm}
\end{subfigure}
\begin{subfigure}{0.19\textwidth}
    \caption{Proposed SDP, 60 GHz}
\includegraphics[width=\textwidth,height=0.7\textwidth]{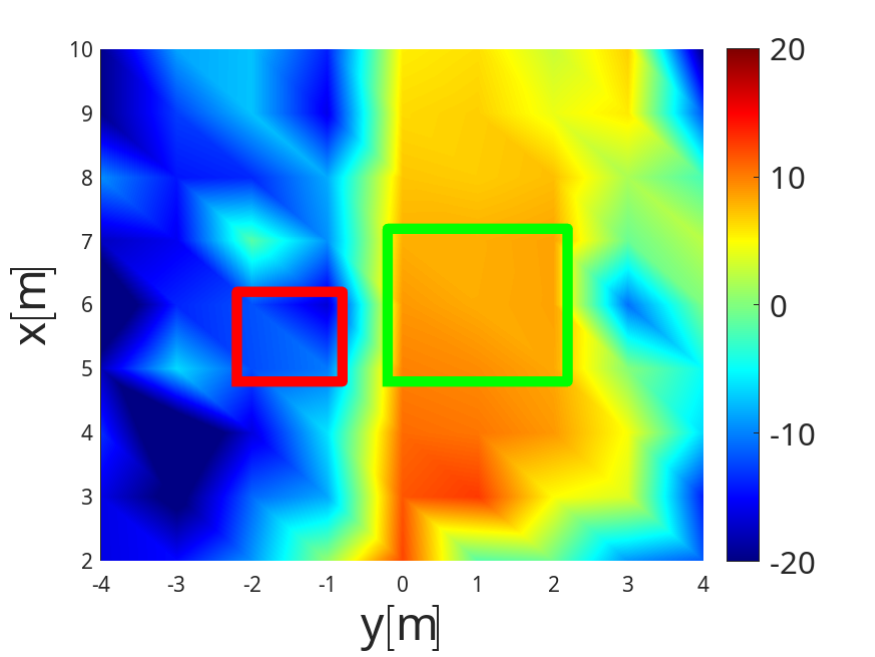}
    \label{fig: Proposed_60}
    \vspace{-5 mm}
\end{subfigure}
\begin{subfigure}{0.19\textwidth}
    \caption{Proposed scalable, 60 GHz}
    \includegraphics[width=\textwidth,height=0.7\textwidth]{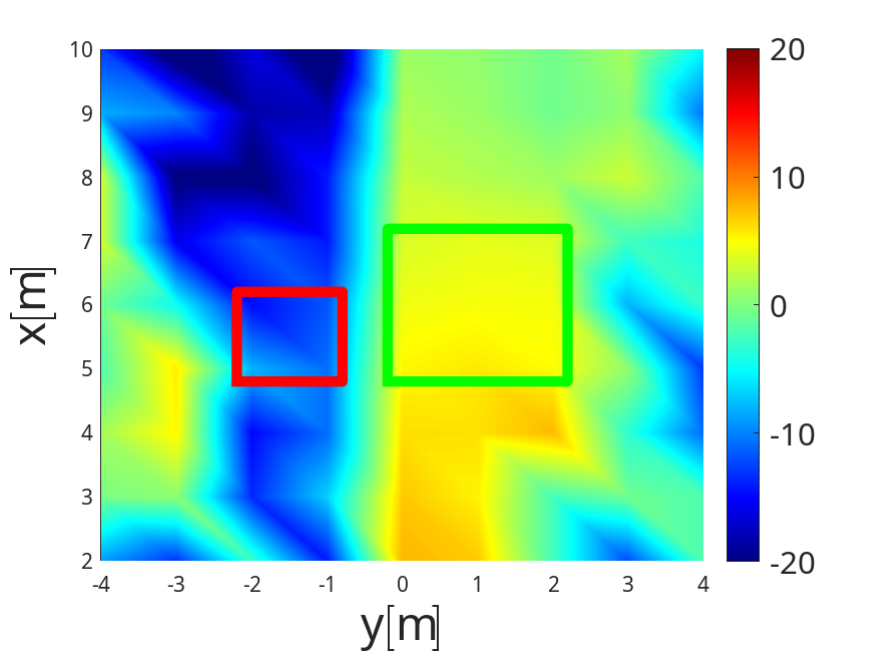}
    \label{fig: Proposed_S_60}
    \vspace{-5 mm}
\end{subfigure}
\begin{subfigure}{0.19\textwidth}
    \caption{Benchmark 1, 60 GHz}
   \includegraphics[width=\textwidth,height=0.7\textwidth]{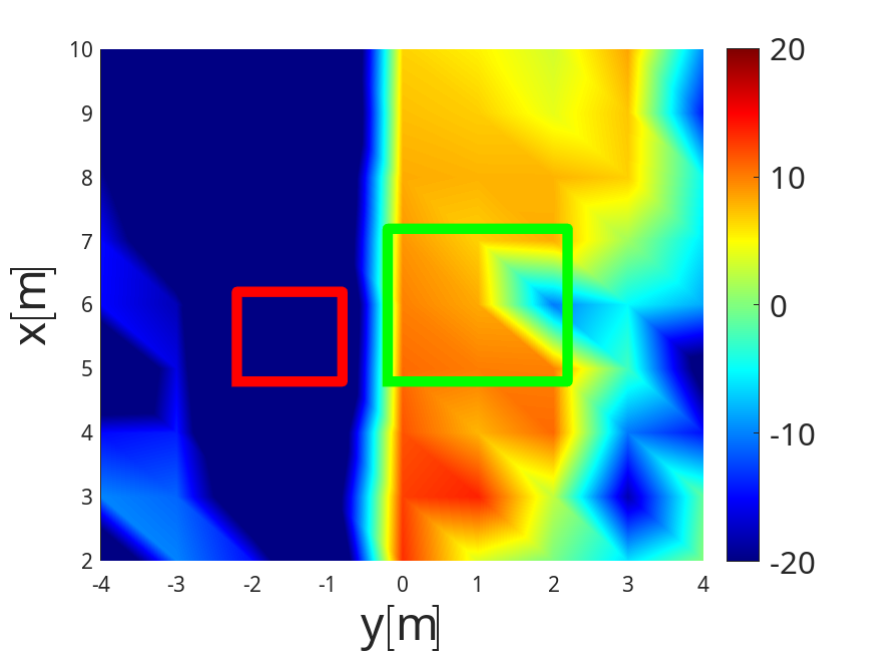}
    \label{fig: Benchmark 1_60}
    \vspace{-5 mm}
\end{subfigure}
\begin{subfigure}{0.19\textwidth}
    \caption{Benchmark 2, 60 GHz}
    \includegraphics[width=\textwidth,height=0.7\textwidth]{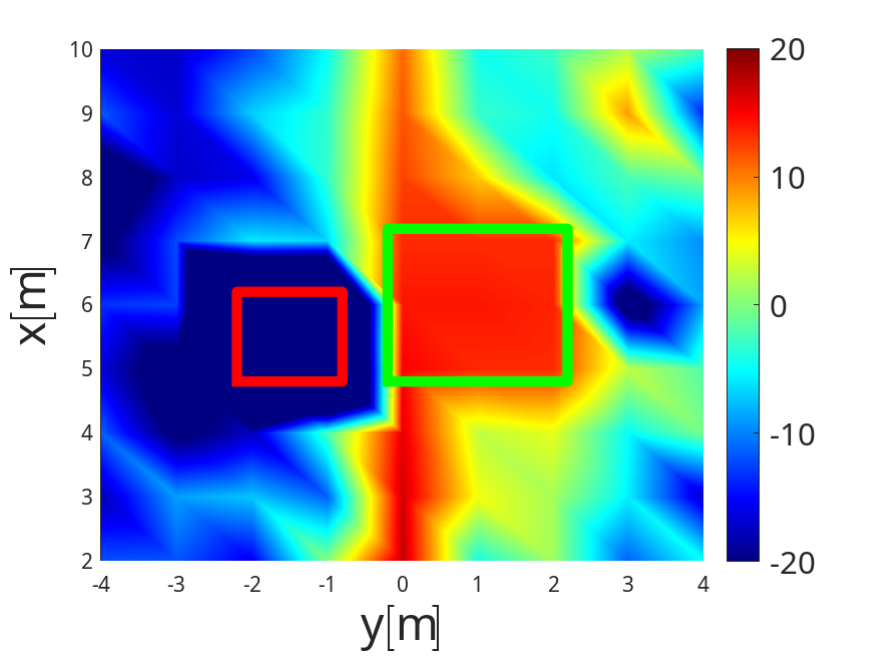}
    \label{fig: benchmark 2_60}
    \vspace{-5 mm}
\end{subfigure}
\begin{subfigure}{0.19\textwidth}
    \caption{Benchmark 3, 60 GHz}
    \includegraphics[width=\textwidth,height=0.7\textwidth]{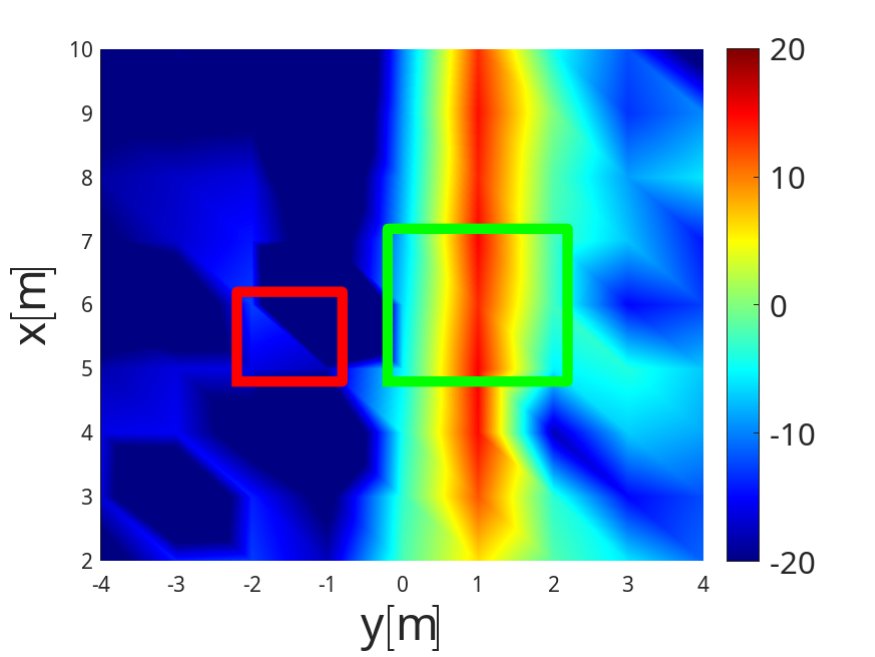}
    \label{fig: benchmark 3_60}
    \vspace{-5 mm}
\end{subfigure}
\begin{subfigure}{0.19\textwidth}
    \caption{Proposed SDP, 64 GHz}
\includegraphics[width=\textwidth,height=0.7\textwidth]{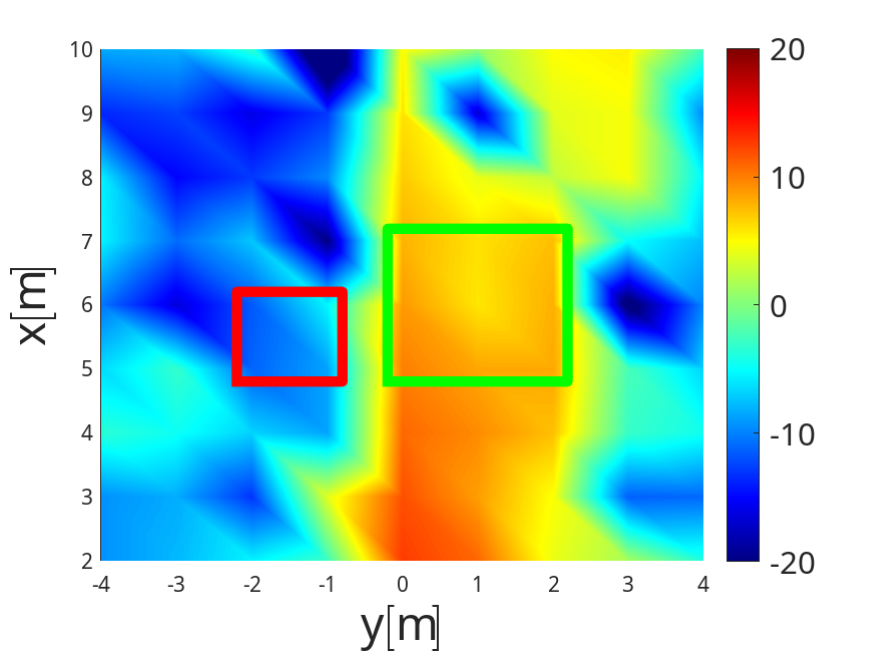}
    \label{fig: proposed_64}
    \vspace{-5 mm}
\end{subfigure}
\begin{subfigure}{0.19\textwidth}
    \caption{Proposed scalable, 64 GHz}
    \includegraphics[width=\textwidth,height=0.7\textwidth]{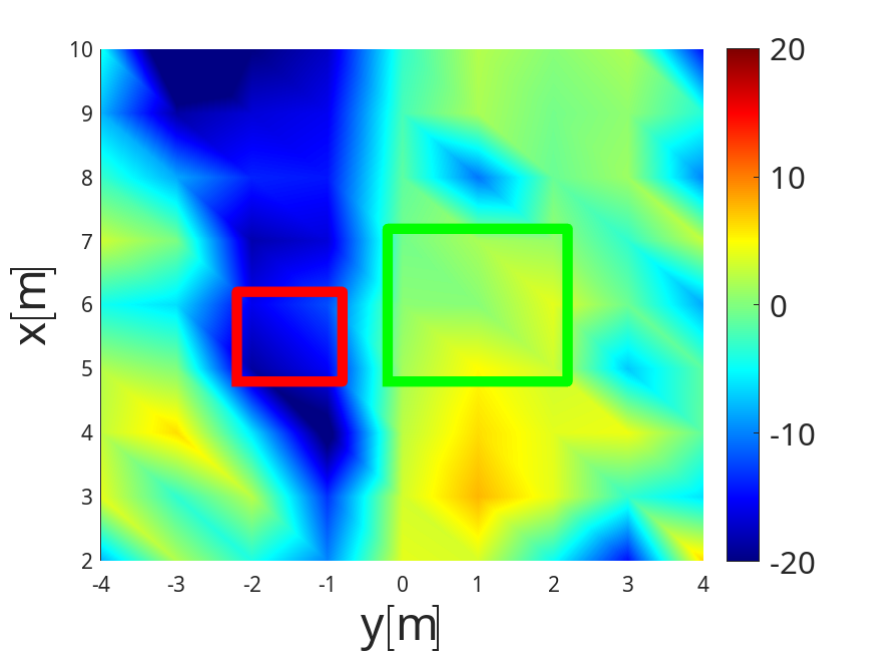}
    \label{fig: Proposed_S_64}
    \vspace{-5 mm}
\end{subfigure}
\begin{subfigure}{0.19\textwidth}
    \caption{Benchmark 1, 64 GHz}
   \includegraphics[width=\textwidth,height=0.7\textwidth]{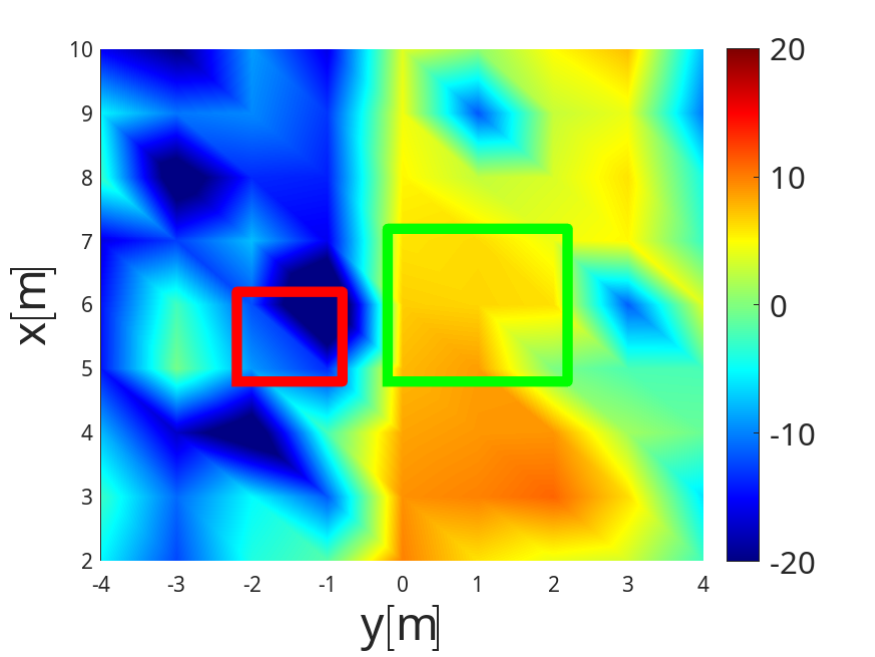}
    \label{fig: benchmark 1_64}
    \vspace{-5 mm}
\end{subfigure}
\begin{subfigure}{0.19\textwidth}
    \caption{Benchmark 2, 64 GHz}
    \includegraphics[width=\textwidth,height=0.7\textwidth]{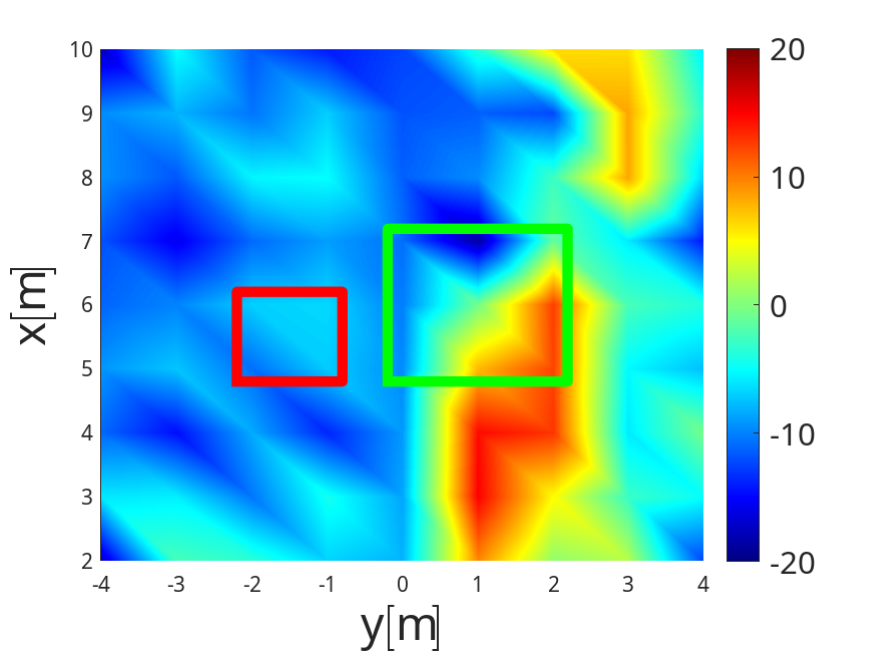}
    \label{fig: benchmark 2_64}
    \vspace{-5 mm}
\end{subfigure}
\begin{subfigure}{0.19\textwidth}
    \caption{Benchmark 3, 64 GHz}
    \includegraphics[width=\textwidth,height=0.7\textwidth]{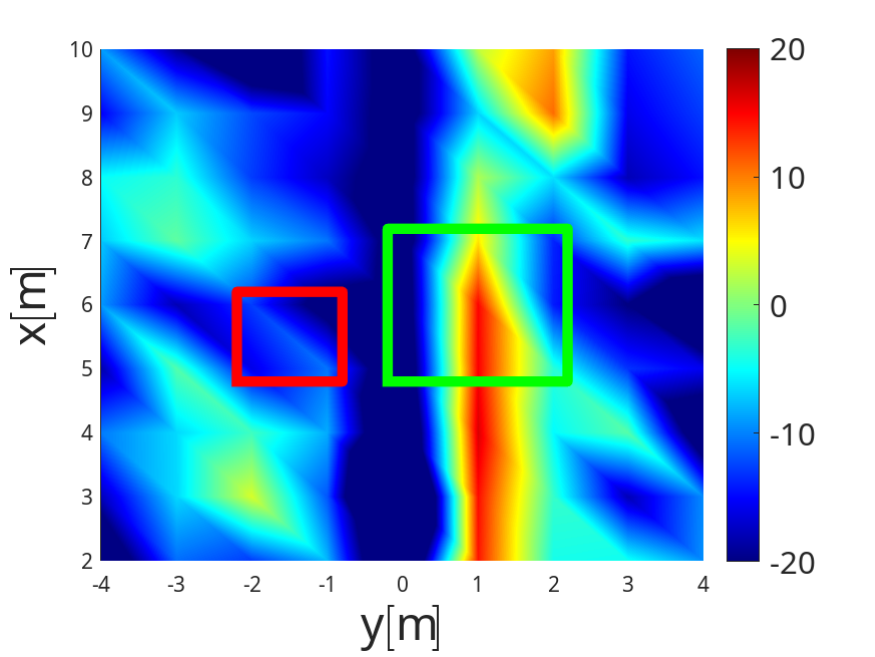}
    \label{fig: benchmark 3_64}
    \vspace{-5 mm}
\end{subfigure}
\caption{\gls{SNR} for different algorithms in different frequencies. The green area represents the user coverage region, while the red area indicates the possible locations of the eavesdropper.}
\vspace{-5 mm}
\label{fig: heat map different frequencies}
\end{figure*}

\subsubsection{Performance comparison}

Fig.~\ref{fig:SR_f} illustrates the secrecy rate (bits/symbol) \gls{wrt} frequency, as defined in \eqref{eq: SNR}, under the assumption of worst-case received power for legitimate users within the region $\Pset_u$ and maximum received power for the eavesdropper within the region $\Pset_e$. As shown in the figure, the minimum secrecy rate achieved by our proposed methods across all frequencies consistently exceeds that of the benchmark methods. While the \gls{SDP}-based approach (Algorithm~\ref{alg:cap}) yields superior performance compared to the scalable low-complexity method (Algorithm~\ref{alg:cap 2}), this improvement comes at the cost of significantly higher computational complexity. Benchmark 3 distributes the reflected power across all frequencies; however, since its design depends on the exact user and eavesdropper locations, its worst-case performance is inferior to the others. Benchmark 2 concentrates power within the target area but only at the center frequency, which results in a secrecy rate that degrades as the signal frequency deviates from the center. Finally, Benchmark 1 optimizes the phase shifts based on the signal frequency and the considered areas, but it neglects the frequency-dependent phase response of individual \gls{RIS} elements (similar to the other benchmarks).

To illustrate how the \gls{SNR} is distributed spatially at a specific frequency, we present heat maps of the \gls{SNR} for all five methods in Fig.~\ref{fig: heat map different frequencies}, corresponding to the carrier frequency $f_k = \{56, 60, 64\}$~GHz. The coverage area and the eavesdropper area are indicated by green and red rectangles, respectively. In Figs.~\ref{fig: Benchmark 3_56}, \ref{fig: benchmark 3_60}, and \ref{fig: benchmark 3_64}, corresponding to Benchmark 3, the \gls{SNR} within the desired area is non-uniform, as this method focuses only on frequency and not on spatial regions. This non-uniformity leads to degraded performance. Figs.~\ref{fig: Benchmark 2_56} and \ref{fig: benchmark 2_64} show the result of Benchmark 2, where noticeable beam splitting occurs due to its optimization being limited to the center frequency which is plotted in Fig.~\ref{fig: benchmark 2_60}. As a result, the \gls{SNR} in the eavesdropper area increases. The outcome for Benchmark 1 is illustrated in Figs.~\ref{fig: Benchmark 1_56}, \ref{fig: Benchmark 1_60}, and \ref{fig: benchmark 1_64}. Since it ignores the frequency-dependent phase response of each \gls{RIS} element, the \gls{SNR} decreases in parts of the desired area and, conversely, increases in the eavesdropper area, thereby reducing its overall performance. In contrast, the proposed \gls{SDP}-based method (Figs.~\ref{fig: Proposed_56}, \ref{fig: Proposed_60}, and \ref{fig: proposed_64}) effectively distribute the \gls{SNR} uniformly across the desired area while simultaneously minimizing it in the eavesdropper region, demonstrating its superior spatial and frequency-aware design.  The proposed scalable algorithm achieves a similar spatial distribution, as illustrated in Figs.~\ref{fig: Proposed_S_56}, \ref{fig: Proposed_S_60}, and \ref{fig: Proposed_S_64}. However, it is evident that the received \gls{SNR} in the \gls{MU} area is noticeably lower than that of the \gls{SDP} approach, although its execution time is significantly shorter.
\section{Experimental Results for Wideband LC-RIS}
\label{sec: Experimental Results}

In this section, we evaluate the proposed algorithm using a small-scale proof-of-concept implementation of an \gls{LC}-\gls{RIS}. Our focus is to establish a minimal experimental setup that enables us to study the frequency-dependent properties of the \gls{LC}-\gls{RIS} and evaluate the effectiveness of the proposed algorithm in a real-world setting.

\begin{figure}
    \centering
    \includegraphics[width=0.3\textwidth]{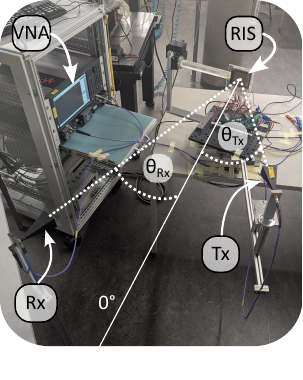}
    \vspace{-5 mm}
    \caption{Experimental setup for measuring the \gls{SNR} in different frequencies at a fixed location. For more details of the setup, you can see \cite[Fig.~5]{neuder2023architecture}.}
    \label{fig:LC_setup}
    \vspace{-5mm}
\end{figure}

\begin{figure*}
\begin{subfigure}{0.3\textwidth}
    \centering
    \includegraphics[width=1\textwidth]{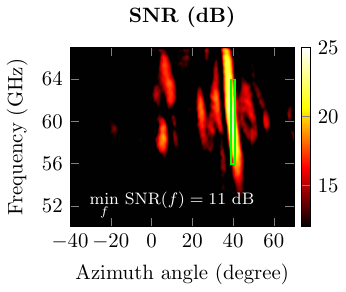}
    \caption{Measurement results for the benchmark 1 phase shift optimization show that the reflected power at 56 GHz and 64 GHz is negligible. This algorithm only focuses on maximizing the \gls{Rx} power in 60 GHz.}
    \label{fig:experimental_benchmark}
\end{subfigure}
\hfill
\begin{subfigure}{0.3\textwidth}
    \centering
    \includegraphics[width=1\textwidth]{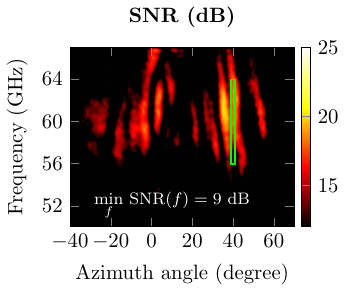}
    \caption{Measurement results for the benchmark 2 phase shift optimization show that the reflected power at 56 GHz and 64 GHz is not good. This algorithm overlooks the frequency dependency in each element.}
    \label{fig:experimental_benchmark 2}
\end{subfigure}
\hfill
\begin{subfigure}{0.3\textwidth}
    \centering
    \includegraphics[width=1\textwidth]{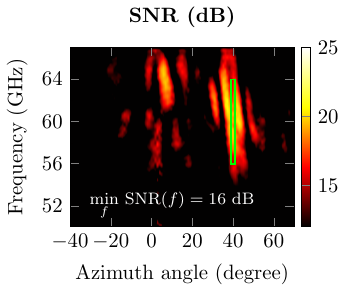}
    \caption{Measurement results for the proposed phase shift optimization show that the reflected power is approximately equal from 56 GHz to 64 GHz, as its design is according to the frequency dependency in elements.}
    \label{fig:experimental_proposed}
\end{subfigure}
\caption{A comparison between the three algorithms is conducted when there are different conditions for phase shift designs. In this analysis, the targeted angle for all methods is $40^\circ$. The green area represents the \gls{Rx} location for the entire frequency band of interest.}
%\vspace{-5mm}
\end{figure*}

\subsection{Experimental Setup}
Our experimental \gls{LC}-\gls{RIS} prototype consists of a $30 \times 25$ grid, totaling 750 elements. To achieve 1D beamforming, a uniform bias voltage is supplied to every element within each individual column\footnote{While this column-level control is constrained by the hardware of our current testbed, our proposed algorithm remains fully adaptable to 2D beamforming applications.}. Consequently, the \gls{LC}-\gls{RIS} can actively sweep reflected signals in the azimuthal plane, while the elevation angle is held constant. These bias voltages are driven by a 1 kHz square wave, generated by a Texas Instruments 12-bit DAC60096 evaluation module (EVM) capable of delivering between $\pm$10.5 V. Finally, as shown in Fig.~\ref{fig:LC_setup}, the measurement environment incorporates a pair of MI-Wave 25 dBi V-band horn antennas \cite{delbari2026fast}.
The \gls{Tx} antenna is fixed at an azimuth angle of $\theta_\mathrm{Tx}=-30^\circ$ and positioned 1 m away from the \gls{LC}-\gls{RIS}. Conversely, the \gls{Rx} antenna is mounted on a rotational stage at a fixed radial distance of 55 cm from the surface, allowing for measurements across arbitrary azimuth angles. For the experiments presented herein, the receiver is positioned at a target azimuth angle of $\theta_\mathrm{Rx}=10^\circ$ \gls{wrt} the \gls{RIS}. This results $10^\circ-(-30^\circ)=40^\circ$ beam steering \gls{wrt} the \gls{Tx}'s angle. Finally, all channel measurements were conducted using a Keysight Technologies PNA-X N5247A vector network analyzer.

Because the number of elements in a single row is insufficient to illuminate a broad area, we optimized the phase shifts to target a single focal point. Furthermore, due to hardware constraints regarding the number of available \gls{Rx} devices, our measurements are limited to the received \gls{SNR}. Nevertheless, even with this restricted metric, we demonstrate that accounting for the frequency dependency of individual \gls{RIS} elements significantly impacts overall performance. Specifically, we compare our proposed method against two alternative phase-shifter optimization strategies. Benchmark 1 designs the phase shifts to maximize the \gls{Rx} power strictly at the 60~GHz center frequency. Benchmark 2 optimizes the phase shifts to reflect power toward $40^\circ$ (\gls{wrt} the \gls{Tx}'s angle) across the entire frequency range, but it ignores the frequency-dependent behavior of the \gls{LC}-\gls{RIS} elements. In contrast, our proposed method explicitly optimizes the phase shifters by incorporating the frequency-dependent characteristics of the \gls{LC}-\gls{RIS} elements across the entire desired bandwidth.

\subsection{Experimental Results}
Figs.~\ref{fig:experimental_benchmark}, \ref{fig:experimental_benchmark 2}, and \ref{fig:experimental_proposed} illustrate the heat maps of the reflected power for Benchmark 1, Benchmark 2, and the proposed method, respectively.  In these plots, the $y$-axis represents the operating frequency, while the $x$-axis denotes the azimuth angle. At the target angle of $40^\circ$, the received \gls{SNR} achieved by the proposed method is distributed almost uniformly across the entire desired frequency band. This consistent wideband performance stands in stark contrast to the benchmark schemes, which exhibit significant power fluctuations and signal degradation across the measured frequencies.

Fig.~\ref{fig: experimental result SNR} presents the received \gls{SNR} across different frequencies obtained from our experimental implementation. As observed, while Benchmark 1 achieves the peak performance near 60~GHz, its effectiveness degrades significantly at the band edges (i.e., 56~GHz and 64~GHz). Benchmark 2 maintains adequate performance between 60~GHz and 64~GHz, but it suffers from severe signal degradation near 56~GHz. In contrast, our proposed method successfully distributes the received power almost uniformly across the entire operating band, from 56~GHz to 64~GHz.

\begin{figure}
    \centering
    \includegraphics[width=0.4\textwidth]{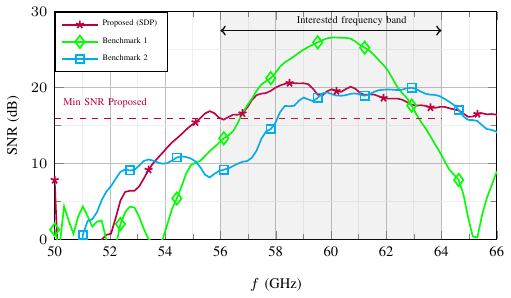}
    \caption{The received \gls{SNR} measurement at the \gls{Rx} location versus frequency.}
    \label{fig: experimental result SNR}
    \vspace{-5mm}
\end{figure}

\section{Conclusion and Future Directions}
\label{sec: Conclusion}
In this paper, we studied the effect of frequency variations on the phase-shift response of \gls{LC}-\gls{RIS} elements. To address the inherent frequency dependency in wideband systems, we proposed two novel algorithms that explicitly account for this behavior in the phase-shift design, with the objective of maximizing the secrecy rate within a targeted spatial area. The first algorithm leverages \gls{SDP}, while the second employs a surrogate function approach, which successfully reduces the computational complexity compared to the \gls{SDP}-based method. Extensive simulation results confirmed the critical importance of incorporating the frequency-dependent characteristics of \gls{LC}-\gls{RIS} elements, demonstrating significant performance gains over baseline approaches that neglect this factor. Specifically, the proposed \gls{SDP}-based and scalable methods achieved secrecy rates of approximately 2 and 1 bits/symbol, respectively, over an 8 GHz bandwidth at a center frequency of 60 GHz. Furthermore, our experimental validations clearly demonstrated that the proposed framework effectively enhances the received power in practical scenarios compared to existing benchmark schemes.

\bibliographystyle{IEEEtran}
\bibliography{References}

\end{document}